\documentclass{llncs}

\usepackage{stackengine}
\usepackage{amssymb}
\usepackage{amsmath,mathrsfs}
\usepackage{enumerate}
\usepackage{standalone}
\usepackage[position=top]{subfig}
\usepackage{graphicx,wrapfig}
\usepackage{tikz}
\usepackage{geometry}
\usepackage{amssymb}
\usetikzlibrary[positioning,decorations.pathmorphing]

\usepackage{thmtools,thm-restate}

\newcommand{\commentout}[1]{}
\newcommand{\heather}[1]{{\color{black} #1}} 
\newcommand{\guarnera}[1]{{\color{black} #1}} 

\newcommand{\smid}{{\mid}}

\newcounter{subfloat}
\renewcommand{\thesubfloat}{\alph{subfloat}}
\newcommand{\image}[2]{%
  \stepcounter{subfloat}%
  \begin{tabular}[t]{@{}c@{}}
  #2 \\
  (\thesubfloat) #1
  \end{tabular}%
}

\makeatletter\@addtoreset{case}{lemma}\makeatother
\makeatletter\@addtoreset{case}{theorem}\makeatother

\title{Obstructions to a small hyperbolicity in Helly graphs}
\author{Feodor F. Dragan \and Heather M. Guarnera}

\begin{document}
\institute{Department of Computer Science, Kent State University, Kent, OH, USA \\
\email{dragan@cs.kent.edu, hmichaud@kent.edu}}

\maketitle

\begin{abstract} 
\heather{
The $\delta$-hyperbolicity of a graph is defined by a simple 4-point condition: for any four vertices $u$, $v$, $w$, and $x$, the two larger of the distance sums $d(u,v) + d(w,x)$,  $d(u,w) + d(v,x)$, and $d(u,x) + d(v,w)$ differ by at most $2\delta \geq 0$.
Hyperbolicity can be viewed as a measure of how close a graph is to a tree metrically; the smaller the hyperbolicity of a graph, the closer it is metrically to a tree. A graph $G$ is Helly if its disks satisfy the Helly property, i.e., every family of pairwise intersecting disks in $G$ has a common intersection. 
It is known that for every graph $G$ there exists the smallest  Helly graph ${\cal H}(G)$ into which $G$ isometrically embeds (${\cal H}(G)$ is called the injective hull of $G$) and
}the hyperbolicity of ${\cal H}(G)$ is equal to the hyperbolicity of $G$. Motivated by this, we investigate structural properties of Helly graphs that govern their hyperbolicity and identify three isometric subgraphs of the King-grid as structural obstructions to a small hyperbolicity in Helly graphs.

\end{abstract}

 %
 %
 \section{Introduction}
 \guarnera{The $\delta$-hyperbolicity of a graph can be viewed as a measure of how close a graph is to a tree metrically; the smaller the hyperbolicity of a graph, the closer it is metrically to a tree.}
 Recent empirical studies indicated that a large number of real-world
 networks, including Internet application networks, web networks, collaboration networks, social networks and biological networks, have small hyperbolicity \cite{AbuAta16,Adcock13,ChFHM,JLB2008,Kennedy13,Montgolfier11,Narayan11,Shavitt08}.
 This motivates much research to understand the structure and characteristics
 of hyperbolic graphs \cite{AbuAta16,Bandelt03,BrKoMo2001,ChChHiOs2014,ChChPaPe2014,Chepoi08,Chepoi17,Coudert14,Koolen02,VS2014,Wu11}, as well as
 algorithmic implications of small hyperbolicity \cite{ChChPaPe2014,Chepoi08,Chepoi17,Chepoi12,Chepoi07,DGKMY2015,EKS2016,Gavoille05,Krauthgamer06,VS2014}.
 One aims at developing approximation algorithms for certain optimization problems whose approximation factor depends only on the hyperbolicity of the input graph.
 To the date such approximation algorithms exist for radius and diameter \cite{Chepoi08}, minimum ball covering \cite{Chepoi07}, $p$-centers \cite{EKS2016}, sparse additive spanners \cite{Chepoi12}, the Traveling Salesmen Problem \cite{Krauthgamer06}, to name a few, which all have an approximation ratio that depends only on the hyperbolicity of the input graph.
 \heather{
 Notably, there is a quasilinear time algorithm~\cite{EKS2016} for the $p$-center problem with additive error at most $3\delta$,
 whereas in general it is known~\cite{Hsu79} that determining an $\alpha$-approximate solution to $p$-centers is NP-hard whenever $\alpha < 2$.
 In another example,}
 there is a linear time algorithm~\cite{Chepoi08} that for any graph $G$ finds a vertex $v$ with eccentricity at most $rad(G)+5\delta$ (almost the radius of $G$) and a pair of vertices $u,v$ such that the distance between $u$ and $v$ is at most $diam(G)-2\delta$ (almost the diameter of $G$), where $\delta$ is the hyperbolicity parameter of $G$.

 In this paper, we are interested in understanding what structural properties of graphs govern their hyperbolicity and in identifying structural obstructions to a small hyperbolicity. It is a well-known fact that the treewidth of a graph $G$ is always greater than or equal to the size of the largest square grid minor of $G$. Furthermore, in the other direction, the celebrated grid minor theorem by Robertson and Seymour \cite{RobSey86} says that there exists a function $f$ such that the treewidth is at most $f(r)$ where $r$ is the size of the largest square grid minor. To the date the best bound on $f(r)$ is $O(r^{98+o(1)})$: every graph of treewidth larger than $f(r)$ contains an $(r \times r)$ grid as a minor~\cite{ChChSTOC2014}.
 Can similar ``obstruction'' results be proven for the hyperbolicity parameter?

 We show in this paper that the {\em thinness of metric intervals} governs the hyperbolicity of a Helly graph and that {\em three isometric subgraphs of the King-grid} are the only obstructions to a small hyperbolicity in Helly graphs. Our interest in Helly graphs (the graphs in which disks satisfy the Helly property) stems from the following two facts. We formulate them in the context of graphs although they are true for any metric space.
 For every graph $G$ there exists the smallest Helly graph ${\cal H}(G)$ into which $G$ isometrically embeds; ${\cal H}(G)$ is called the {\em injective hull} of $G$ \cite{Dr,Is}.
 If $G$ is a $\delta$-hyperbolic graph, then ${\cal H}(G)$ is also $\delta$-hyperbolic and every vertex of ${\cal H}(G)$ is within distance $2\delta$ from some vertex of $G$ \cite{Lang13}.
 Thus, from our main result for Helly graphs (see Theorem \ref{general-forbidden}), one can state the following.
 \begin{theorem}
 An arbitrary graph $G$ has hyperbolicity at most $\delta$ if and only if its injective hull ${\cal H}(G)$ contains \\
  - no $H_2^{\delta}$,  when $\delta$ is an integer, \\
  - neither $H_1^{\delta+\frac{1}{2}}$ nor  $H_3^{\delta-\frac{1}{2}}$, when $\delta$ is a half-integer,  \\
 from Fig. \ref{fig:combinedFrames-hellified} as an isometric subgraph.
 \end{theorem}

 \begin{figure}[htb]
 \vspace*{-9mm}
 \begin{center}
    \footnotesize
    \stackunder[5pt]
    {\includegraphics[scale=.45]{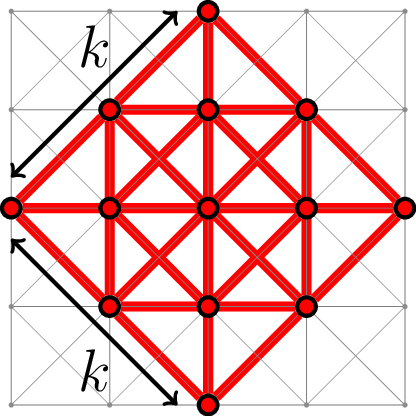}}{$H_1^{k}$}%
    \quad
    \footnotesize
    \stackunder[5pt]
    {\includegraphics[scale=.45]{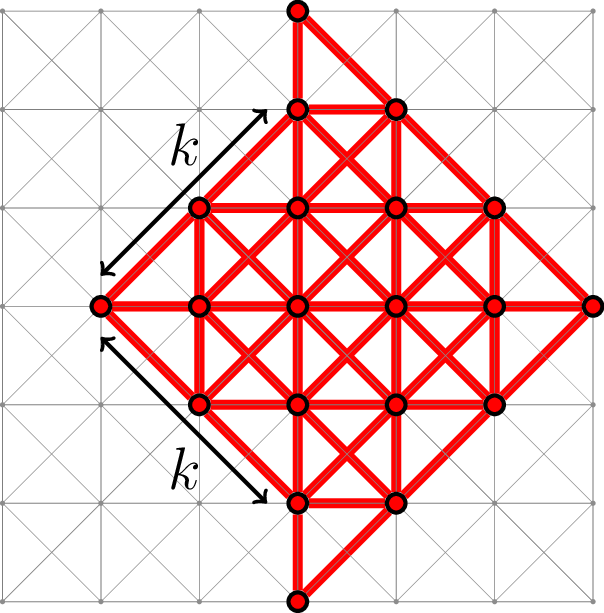}}{$H_2^{k}$}%
    \quad
    \footnotesize
    \stackunder[5pt]
    {\includegraphics[scale=.45]{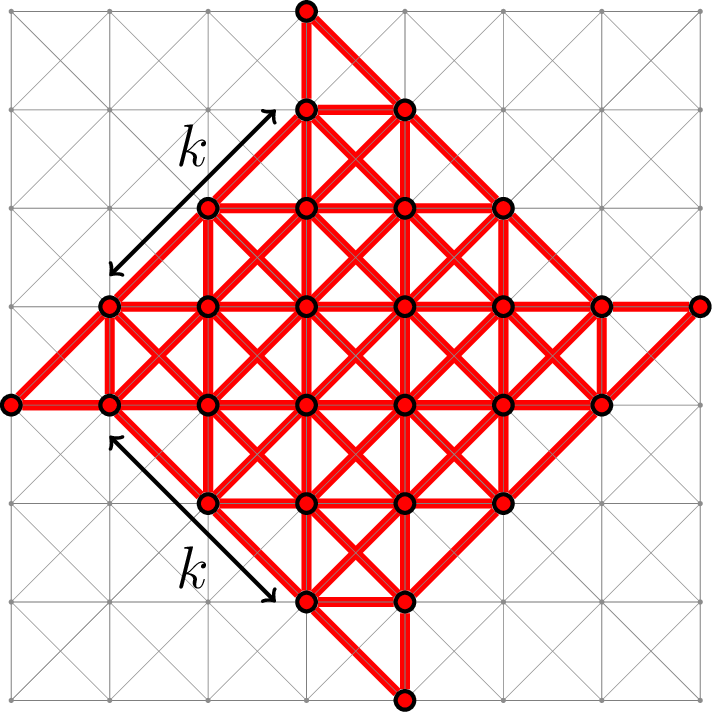}}{$H_3^{k}$}
  \caption{
  \heather{
  Graphs $H_1^{k}$, $H_2^{k}$, and $H_3^{k}$ in red, where $k=2$.
  Each graph is shown isometrically embedded into the King-grid, which is a strong product of two paths and is a particular Helly graph.}
  }
  \end{center} \vspace*{-.5cm}
  \end{figure} \label{fig:combinedFrames-hellified}


The injective hull here can be viewed as playing a similar role as the minors in the  grid minor theorem for treewidth by Robertson and Seymour.
Helly graphs play a similar role for the hyperbolicity as chordal graphs play for the treewidth.
Note that each of the graphs  $H_1^k$, $H_2^k$, $H_3^k$ contains a square grid of side $k$ (see Fig. \ref{fig:combinedFrames-hellified}) as an isometric subgraph. Thus, if the hyperbolicity of a Helly graph is large then it has a large square grid as an isometric subgraph. 
This result (along with a connection between the treewidth and the hyperbolicity established in  \cite{Coudert16}\footnote{In fact, \cite{Coudert16} establishes a relation between the treelength and the treewidth of a graph but according to \cite{Chepoi08} the hyperbolicity and the treelength are within a factor of $O(\log n)$ from each other.}) calls for an attempt to develop a theory similar to the bidimensionality theory (see survey \cite{Demaine2008} and papers cited therein). 
The bidimensionality theory builds on the graph minor theory of Robertson and Seymour by extending the mathematical results and
building new algorithmic tools.  Using algorithms for graphs of bounded treewidth as sub-routines (see also an earlier paper \cite{Eppstein2000}), it provides general techniques for designing efficient fixed-parameter  algorithms  and  approximation  algorithms 
for   NP-hard   graph   problems   in   broad   classes   of graphs.  This theory applies to graph problems that are "bidimensional``
in the sense that (1) the solution value for  the $k\times k$ grid  graph  (and  similar  graphs)  grows with $k$,  typically  as  $\Omega(k^2)$,  and  (2)  the  solution  value goes down when contracting edges, and optionally when deleting edges, in the graph. Examples of such problems include  feedback  vertex  set,  vertex  cover,  minimum maximal matching, face cover, a series of vertex-removal parameters, (edge)  dominating  set, connected (edge)   dominating   set, unweighted TSP tour, and chordal completion (fill-in). We are currently investigating this direction.

\commentout{For general graphs, finding ${\cal H}(G)$ is $O((diam(G)+1)^n)$ (c.f. \cite{Lang13} for details on the construction).
Given a Helly graph $H$, it is $O(n^4)$ to naively check if $H$ contains one of the three isometric subgraphs, by calculating the shortest path between all pairs and comparing the pairwise-distances between any four vertices in $H$.
However, the existence of the injective hull can be used to prove properties of importance in $G$ without constructing ${\cal H}(G)$. For example,
in \cite{Chepoi17} the existence of ${\cal H}(G)$ is used to show that any finite subset $X$ of vertices in a locally finite $\delta$-hyperbolic graph $G$ admits a core, that is, a disk $D(m,4\delta$) centered at vertex $m$,
which intercepts all shortest paths between at least one half of all pairs of vertices of $X$.
}

Previously, it was known that the hyperbolicity of median graphs is controlled by the size of isometrically embedded square grids (see \cite{Bandelt08,ChChHiOs2014}), and recently \cite{ChChHiOs2014} showed that the hyperbolicity of weakly modular graphs (a far reaching superclass of the Helly graphs) is controlled by the sizes of metric triangles and isometric square grids: if $G$ is a weakly modular graph in which any metric triangle is of side at most $\mu$ and any isometric square grid contained in $G$ is of side at most $\nu$, then  $G$ is $O(\nu + \mu)$-hyperbolic. Recall that three vertices $x,y,z$ of a graph form a metric triangle if for each vertex $v\in \{x,y,z\}$, any two shortest paths connecting it with the two other vertices from $\{w,y,z\}$ have only $v$ in common.
\heather{
Projecting this general result to Helly graphs (where $\mu\le 1$) one gets only that every Helly graph with isometric grids
of side at most $\nu$ is  $c\nu$-hyperbolic with a constant $c$ larger than 1 (about 8).
}

Injective hulls of graphs were recently used in \cite{Chepoi17} to prove a conjecture by Jonckheere et al.~\cite{JoLoBoBa} that real-world networks with small hyperbolicity have a core congestion.
It was shown \cite{Chepoi17} that any finite subset $X$ of vertices in a locally finite $\delta$-hyperbolic graph $G$ admits a disk $D(m,4\delta$) centered at vertex $m$,
which intercepts all shortest paths between at least one half of all pairs of vertices of $X$.

There has also been much related work on the characterization of $\delta$-hyperbolic graphs via forbidden isometric subgraphs - particularly, when $\delta=\frac{1}{2}$.
Koolen and Moulton \cite{Koolen02} provide such a characterization for $\frac{1}{2}$-hyperbolic bridged graphs via six forbidden isometric subgraphs.
Bandelt and Chepoi \cite{Bandelt03} generalize these results to all $\frac{1}{2}$-hyperbolic graphs via the same forbidden isometric subgraphs and the property that all disks of $G$ are convex.
Additionally, Coudert and Ducoffe \cite{Coudert14} prove that a graph is $\frac{1}{2}$-hyperbolic if and only if every graph power $G^i$ is $C_4$-free for $i\geq 1$, and one additional graph is $C_4$-free.
Brinkmann et al. \cite{BrKoMo2001} characterize $\frac{1}{2}$-hyperbolic chordal graphs via two forbidden isometric subgraphs. 
Wu and Zhang \cite{Wu11} prove that a 5-chordal graph is $\frac{1}{2}$-hyperbolic if and only if it does not contain six isometric subgraphs. 
Cohen et al. \cite{Cohen14} prove that a biconnected outerplanar graph is $\frac{1}{2}$-hyperbolic if and only if either it is isomorphic to $C_5$ or it is chordal and does not contain a forbidden subgraph.
\heather{
  We present a characterization of $\delta$-hyperbolic Helly graphs, for every $\delta$, with three forbidden isometric subgraphs. Further characterizations of Helly graphs with small hyperbolicity constant are deduced from our main result.
}

 %
 %
 \section{Preliminaries}
 \heather{We use the terminology and definitions as described in standard graph theory textbooks \cite{Diestel,West}.}
 All graphs $G=(V,E)$ appearing here are connected, finite, unweighted, undirected, loopless and without multiple edges.
 The {\em length of a path} from a vertex $v$ to a vertex $u$ is the number of edges in the path. The {\em distance} $d_G(u,v)$ between vertices $u$ and $v$ is the length of a shortest path connecting $u$ and $v$ in $G$. We omit the subscript when $G$ is known by context.  For a subset $A\subseteq V$, a
 subgraph $G(A)$ of a graph $G$ {\em induced} by $A$ is defined as $G(A)=(A,E')$ where
 $uv \in E'$ if and only if $u,v \in A$ and $uv \in E$. An induced subgraph $H$ of $G$ is {\em isometric} if the distance
 between any pair of vertices in $H$ is the same as their distance in $G$.
 The {\em $k$-th power} $G^k$ of $G$ is
 defined as $G^k =(V,E')$ where $E'=\{uv : u,v \in V$ and $d(u,v) \leq k\}$. A {\em disk} $D(v,r)$
 of a graph $G$ centered at a vertex $v \in V$ and with radius $r$ is the set of all vertices with
 distance no more than $r$ from $v$ (i.e., $D(v,r)=\{u\in V: d_G(v,u) \leq r \}$). For any two
 vertices $u$, $v$ of $G$, $I(u,v)= \lbrace z\in V:d(u,v)=d(u,z)+d(z,v) \rbrace$ is the
 (metric) {\sl interval} between $u$ and $v$, i.e., all vertices that lay on shortest
 paths between $u$ and $v$.

A family $\cal F$ of sets $S_i$ has the {\em Helly property} if for every subfamily ${\cal F'}$ of $\cal F$ the
following holds: if the elements of ${\cal F'}$ pairwise intersect, then the intersection of all elements of ${\cal F'}$
is also non-empty. A graph is called {\em Helly} if its family of
 all disks ${\cal D}(G)=\{D(v,r): v\in V, r\in {N}\}$ satisfies the Helly property. Note that two disks $D(v,p)$ and $D(u,q)$ intersect each other if and only if $d_G(u,v)\le p+q$. Two disks
 $D(v,p)$ and $D(u,q)$ of $G$ are said to  {\em see} each other, sometimes also referred to as {\em touching} each other, if they intersect or there is an edge in $G$ with one end in $D(v,p)$ and other end in $D(u,q)$ (equivalently, if $d_G(u,v)\le p+q+1$).
 The {\em strong product} of a set of graphs $G_i$ for $i=1,2,...,k$ is the graph $\boxtimes_{i=1}^kG_i$ whose vertex set
 is the Cartesian product of the vertex sets $V_i$, and there is an edge between vertices $a=(a_1,a_2,...,a_k)$ and $b=(b_1,b_2,...,b_k)$ if and only if
 $a_i$ is either equal or adjacent to $b_i$ for $i=1,2,...,k$. A {\em King-grid} is a strong product of two paths.
 King-grids form a natural subclass of  Helly graphs. 

 The following lemma will be frequently used in this paper.  It is true for a larger family of pseudo-modular graphs
 but we will use it in the context of Helly graphs. Pseudo-modular graphs are exactly the graphs where each family of three pairwise intersecting disks has a common intersection \cite{Bandelt86Modular}. Clearly, Helly graphs is a subclass of pseudo-modular graphs. 

 \begin{center}
 \centering\image{}{
   \includegraphics[scale=.55]{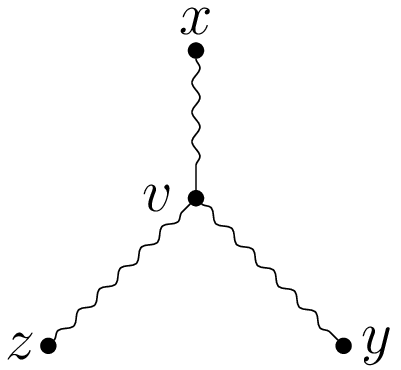}
 }\quad
 \image{}{
    \includegraphics[scale=.55]{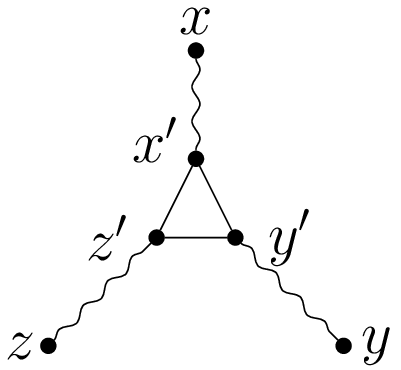}
 }
 \captionof{figure}{Vertices $x,y,z$ and three shortest paths connecting them in pseudo-modular graphs.}
 \label{fig:pseudomodular}
 \end{center}
 \setcounter{subfloat}{0}

\begin{lemma} [\cite{Bandelt86Modular}] \label{centroids}
   For every three vertices $x$, $y$, $z$ of a pseudo-modular graph $G$ there exist three shortest paths $P(x,y)$, $P(x,z)$,
   $P(y,z)$ connecting them such that either (1) there is a common vertex $v$ in $P(z,y) \cap P(x,z) \cap P(x,y)$
   or (2) there is a triangle $\bigtriangleup (x',y',z')$ in $G$ with edge $z'y'$ on $P(z,y)$,
   edge $x'z'$ on $P(x,z)$ and edge $x'y'$ on $P(x,y)$  (see Fig. \ref{fig:pseudomodular}). Furthermore, (1) is true if and only if $d(x,y)=p+q$, $d(x,z)=p+k$ and $d(y,z)=q+k$, for some $k,p,q\in N$, and (2) is true if and only if $d(x,y)=p+q+1$, $d(x,z)=p+k+1$ and $d(y,z)=q+k+1$, for some $k,p,q\in N$.
\end{lemma}

 We are  interested in hyperbolic graphs (sometimes referred to as graphs with a negative
 curvature).  $\delta$-Hyperbolic metric spaces have been defined by Gromov
 \cite{Gromov87} in 1987 via a simple 4-point condition: for any four
 points $u,v,w,x$, the two larger of the distance sums
 $d(u,v)+d(w,x), d(u,w)+d(v,x), d(u,x)+d(v,w)$ differ by at most
 $2\delta \geq 0$. They play an important role in geometric group theory and in
 the geometry of negatively curved spaces, and have recently become of
 interest in several domains of computer science, including
 algorithms and networking. A connected
 graph $G=(V,E)$ equipped with standard graph metric $d_G$ is
 $\delta$-{\it hyperbolic} if the metric space $(V,d_G)$ is
 $\delta$-hyperbolic. The smallest value $\delta$ for which $G$  is $\delta$-hyperbolic is
 called the {\em hyperbolicity} of $G$ and denoted by $hb(G)$.
 Let also $hb(u,v,w,x)$ ($u,v,w,x\in V$) denote one half of the difference between the two larger distance sums
 from $d(u,v)+d(w,x), d(u,w)+d(v,x), d(u,x)+d(v,w)$.
 The Gromov product of two vertices $x,y \in V$ with respect to a third vertex $z \in V$ is
 defined as $(x|y)_z = \frac{1}{2}(d(x,z) + d(y,z) - d(x,y))$.
 The focus of this paper is primarily on {\em $\delta$-hyperbolic Helly graphs} (i.e., graphs which
 satisfy the Helly property as well as have $\delta$ hyperbolicity).

Using the Gromov product, we can reformulate Lemma \ref{centroids}. 
Since $d(x,y)=(x|z)_y + (z|y)_x$, it is easy to check that for any three vertices $x,y,z$ of an arbitrary graph,
either all products  $(y|z)_x$,  $(y|x)_z$,  $(x|z)_y$ are integers or all are half-integers.

\begin{restatable}{lemma}{fnew}\label{fnew}
   For every three vertices $x$, $y$, $z$ of a pseudo-modular graph $G$ there exist three shortest paths $P(z,y)$, $P(x,z)$,
   $P(x,y)$ connecting them such that either (1) there is a common vertex $v$ in $P(z,y) \cap P(x,z) \cap P(x,y)$
   or (2) there is a triangle $\bigtriangleup (x',y',z')$ in $G$ with edge $z'y'$ on $P(z,y)$,
   edge $x'z'$ on $P(x,z)$ and edge $x'y'$ on $P(x,y)$  (see Fig. \ref{fig:pseudomodular}). Furthermore, (1) is true if and only if
   $(x|y)_z$ is an integer and $(x|y)_z=d(z,v)$, and (2) is true if and only if $(x|y)_z$ is a half-integer and $\lfloor (x|y)_z\rfloor=d(z,z')$.
\end{restatable}
\begin{proof}
\heather{
  Let $\alpha_z=(x \smid y)_z$, $\alpha_x =(z|y)_x$, and $ \alpha_y = (z|x)_y$.}
  We have $d(z,x)-\alpha_z= d(z,x)- \frac{1}{2}(d(x,z)+d(y,z)-d(x,y))=\frac{1}{2}(d(x,z)+d(x,y)-d(y,z))=(z|y)_x$. Similarly, $(z|x)_y=d(z,y)-\alpha_z$.
\heather{Therefore, $\alpha_x$ and $ \alpha_y$ are integers if and only if $\alpha_z$ is an integer.}

\heather{
  Let there be a common vertex $v$ in $P(z,y) \cap P(x,z) \cap P(x,y)$.
  By definition, $2\alpha_z = (d(x,v) + d(v,z)) + (d(z,v) + d(v,y)) - (d(x,v) + d(v,y)) = 2d(v,z)$.
  Thus, $d(v,z) = \alpha_z$, and since $d(v,z)$ is an integer then $\alpha_z$ is an integer.
  The converse follows from Lemma~\ref{centroids} for $p=\alpha_x$, $q=\alpha_y$ and $k=\alpha_z$.

  Let there be a triangle $\bigtriangleup (x',y',z')$ in $G$ with edge $z'y'$ on $P(z,y)$,
  edge $x'z'$ on $P(x,z)$ and edge $x'y'$ on $P(x,y)$.
  By definition, $2\alpha_z = (d(x,x') + 1 + d(z',z)) + (d(z,z') + 1 + d(y',y)) - (d(x,x') + 1 + d(y',y)) = 2d(z,z') + 1$.
  Thus, $d(z,z') = \lfloor \alpha_z \rfloor$, and since $2d(z,z') + 1$ is odd then $\alpha_z$ is a half-integer.
  The converse follows from Lemma~\ref{centroids} for
  $p=\lfloor\alpha_x\rfloor$, $q=\lfloor\alpha_y\rfloor$ and $k=\lfloor\alpha_z\rfloor$. \qed

}


\end{proof}

The set $S_k(x,y)=\{z \in I(x,y) : d(z,x) = k \}$ is called a \emph{slice} of the interval from $x$ to~$y$. The diameter of a slice $S_k(x,y)$ is the maximum distance in $G$ between any two vertices of $S_k(x,y)$. An interval $I(x,y)$ is said to be $\tau$-thin if diameters of all slices $S_k(x,y)$, $k\in N$, of it are at most $\tau$. A graph $G$ is said to have $\tau$-thin intervals if all intervals of $G$ are $\tau$-thin. The smallest $\tau$ for which
all intervals of $G$ are $\tau$-thin is called the {\em interval thinness} of $G$ and denoted by  $\tau(G)$. That is,
$$\tau(G)=\max\{d(u,v): u,v\in S_k(x,y), x,y\in V, k\in N\}.$$
The following lemma is a folklore and easy to show using the definition of hyperbolicity.

 \begin{restatable}{lemma}{halfthin} \label{half-thin}
  For any graph $G$, $\tau(G) \leq {2}hb(G)$.
 \end{restatable}
 \begin{proof}
 Consider any interval $I(x,y)$ in $G$ and arbitrary two vertices $u,v\in S_k(x,y)$. Consider the three distance sums
 $S_1=d(x,y)+d(u,v)$, $S_2=d(x,u)+d(y,v)$, $S_3=d(x,v)+d(y,u)$. As $u,v\in S_k(x,y)$, we have $S_2=S_3=d(x,y)\leq S_1$. Hence,
 $2 hb(G)\ge S_1-S_2= d(x,y)+d(u,v)-d(x,y)=d(u,v)$ for any two vertices from the same slice of $G$, i.e., $2 hb(G)\ge \tau(G)$.\qed
 \end{proof}

%

 %
 %
 \section{Thinness of intervals governs the hyperbolicity of a Helly graph} \label{sec:thinness}
 \heather{
 A qualitative relationship between hyperbolicity and thinness of intervals is easy to show even for a more general class of median graphs.
 The true contribution of our paper is more quantitative.
 In fact, we obtain the exact relationship between the two.
 }
 We focus now on demonstrating that the converse of Lemma \ref{half-thin} for Helly graphs is also
 \heather{true such that the value of~$2hb(G)$ is upper bounded by~$\tau(G) + 1$.
 }
 Note that, for general graphs $G$,
 the values of~$\tau(G)$ and~${2}hb(G)$ can be very far from each other. Consider
 an odd cycle with $4k+1$ vertices; each pair of vertices has a unique shortest path, so no two vertices are in the same slice. Thus $\tau(G)=0$ and $2hb(G)=2k$.

  \begin{center}
 \centering\image{}{
   \includegraphics[scale=.57]{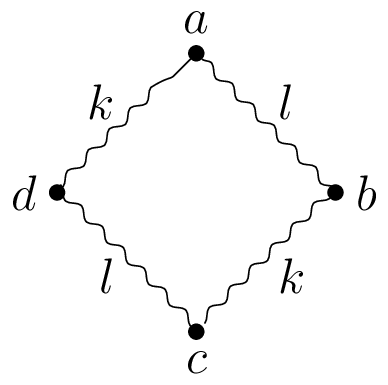}
 }\quad
 \image{}{
   \includegraphics[scale=.57]{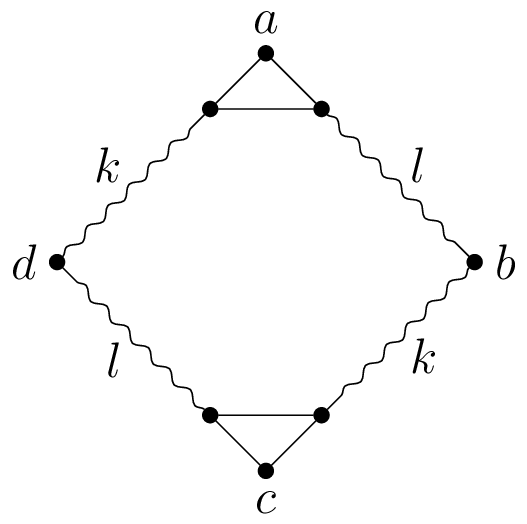}
 }\quad
 \image{}{
   \includegraphics[scale=.57]{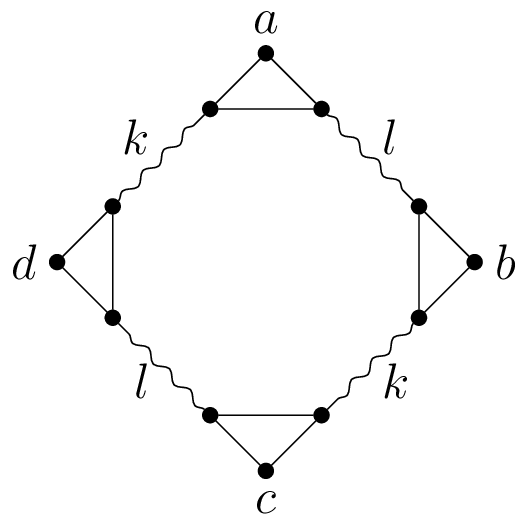}
 }
 \captionof{figure}{$\{a,b,c,d\}$-distance preserving subgraphs.}
 \label{fig:unhellified-frames}
 \end{center}
 \setcounter{subfloat}{0}

We say that a graph $G'=(V',E')$ with $\{a,b,c,d\}\subset V'$ is an {\em $\{a,b,c,d\}$-distance-preserving subgraph} of a graph $G$  if
$d_G(x,y)=d_{G'}(x,y)$ for every pair of vertices $x,y$ from $\{a,b,c,d\}$. On Fig. \ref{fig:unhellified-frames}(c), an  $\{a,b,c,d\}$-distance-preserving subgraph of a graph $G$  with $d_G(a,c)=d_G(b,d)=k+l+3$, $d_G(a,b)=d_G(c,d)=l+2$, and $d_G(b,c)=d_G(d,a)=k+2$ is shown.

  \begin{restatable}{lemma}{hyperbolicitylessthanthinness}\label{hyperbolicity-less-than-thiness}
  \heather{
  For every Helly graph $G$, $2hb(G) \leq \tau(G)+1$. Furthermore, for any Helly graph $G$, $2hb(G)=\tau(G)+1$}
   if and only if $\tau(G)$ is odd and there exists in $G$ an $\{a,b,c,d\}$-distance-preserving subgraph
   \heather{for some $\{a,b,c,d\}$ as depicted} on Fig. \ref{fig:unhellified-frames}(c) with $k=l=\lfloor\frac{\tau(G)}{2}\rfloor$.
 \end{restatable}
 \begin{proof}
 Consider arbitrary four vertices $a,b,c,d$ with $hb(G)=hb(a,b,c,d)=:\delta$ and let $d_G(a,c)+d_G(b,d)\geq d_G(a,b)+d_G(c,d)\ge d_G(a,d)+d_G(b,c)$.
 Let also $\tau:=\tau(G)$.
 \heather{
  We apply Lemma~\ref{centroids} once to vertices $\{a,b,c\}$ and again to vertices $\{a,d,c\}$.
  By Lemma~\ref{centroids}, a set of three vertices with some shortest paths connecting them define either configuration (1) or configuration (2) from  Fig.~\ref{fig:pseudomodular}. Hence, there are three cases, up-to symmetry, to consider. 
  
 }

  \begin{figure}[htb]
     \vspace*{-3mm}
    \begin{center} %
      \begin{minipage}[b]{4cm}
        \begin{center} \includegraphics[height=2.9cm]{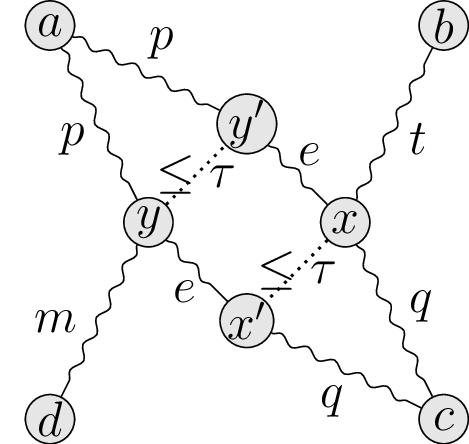}
        \end{center}
        \caption{\label{fig:thinness-to-hyperbolicity-c1} Illustration for Case \ref{case1} for the proof of Lemma \ref{hyperbolicity-less-than-thiness}.} %
      \end{minipage}
      \hspace*{1cm}
      \begin{minipage}[b]{4.1cm}
        \begin{center} \includegraphics[height=3.05cm]{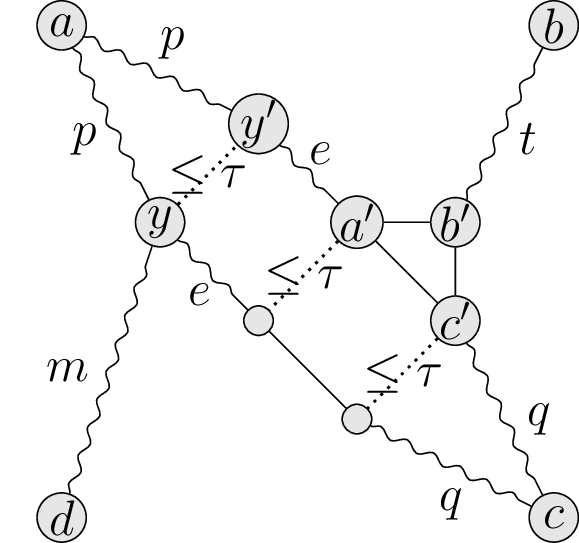}
        \end{center}
        \caption{\label{fig:thinness-to-hyperbolicity-c2} Illustration for Case \ref{case2} for the proof of Lemma \ref{hyperbolicity-less-than-thiness}.} %
      \end{minipage}
      \hspace*{1cm}
      \begin{minipage}[b]{4.1cm}
         \begin{center} \includegraphics[height=3.19cm]{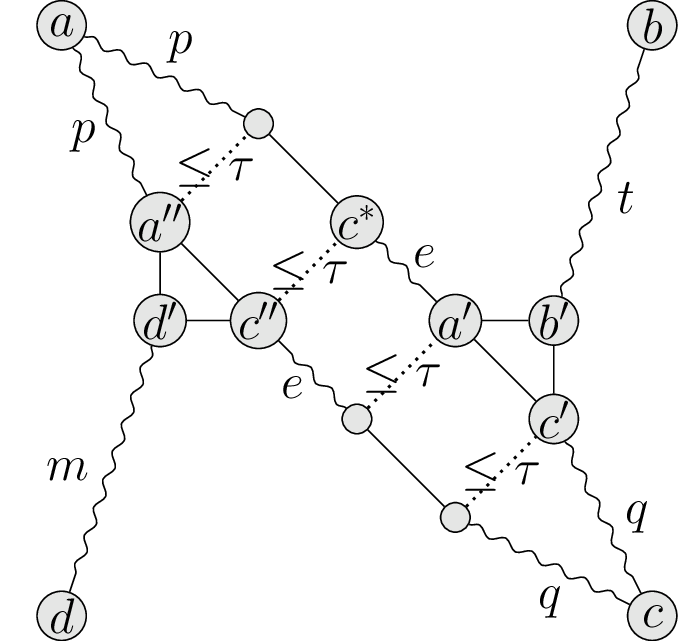}
         \end{center}
         \caption{\label{fig:thinness-to-hyperbolicity-c3} Illustration for Case \ref{case3} for the proof of Lemma \ref{hyperbolicity-less-than-thiness}.}
       \end{minipage}
    \end{center} \vspace*{-3mm}
  \end{figure}

 \begin{case}\label{case1} {\em For vertices $a,b,c$ there are three shortest paths $P(a,b)$, $P(b,c)$, $P_b(a,c)$ that share a common vertex $x$.
 For vertices $a,d,c$ there are three shortest paths $P(a,d)$, $P(d,c)$, $P_d(a,c)$ that share a common vertex
 $y$.
 \heather{
 Note that we use the notation $P_b(a,c)$ and $P_d(a,c)$ here to distinguish between the two shortest $(a,c)$-paths that exist by applying Lemma~\ref{centroids} to vertices $a,b,c$ and again separately to vertices $a,d,c$, respectively.}
 }

This situation is shown on Fig.~\ref{fig:thinness-to-hyperbolicity-c1}. It is unknown if $x$ and $y$ are on the same slice of $I(a,c)$ or not, so we consider vertices $y'\in P_b(a,c)$ and $x'\in P_d(a,c)$ with $d_G(a,y)=d_G(a,y')=:p$ and $d_G(c,x)=d_G(c,x')=:q$. Set also $e:=d_G(y',x)=d_G(y,x')$, $m:=d_G(d,y)$, $t:=d_G(b,x)$ (see Fig.~\ref{fig:thinness-to-hyperbolicity-c1}).
 Vertices $x, x'$ lie on the same slice of $I(a,c)$, as do $y, y'$. Given that intervals of $G$ are $\tau$-thin, we get $2 \delta = d_G(a,c)+d_G(b,d)- (d_G(a,b)+d_G(c,d))\leq p+e+q+t+\tau+e+m -(p+e+t+m+e+q)=\tau$, i.e., \heather{$2\delta\le \tau$.}
 \end{case}

 \begin{case}\label{case2}  {\em For vertices $a,b,c$  there are three shortest paths $P(a,b)$, $P(b,c)$, $P_b(a,c)$ and a triangle $\bigtriangleup (b',c',a')$ in $G$ with edge $a'b'$ on $P(a,b)$, edge $b'c'$ on $P(b,c)$ and edge $a'c'$ on $P_b(a,c)$. For vertices $a,d,c$ there are three shortest paths $P(a,d)$, $P(d,c)$, $P_b(a,c)$ that share a common vertex $y$.}

 This situation is shown on Fig.~\ref{fig:thinness-to-hyperbolicity-c2}.
 \heather{Consider vertex $y' \in P_b(a,c)$ such that $p := d_G(a,y')=d_G(a,y)$, and let now $q := d_G(c,c')$, set $e := d_G(y',a')$,
 and $t := d_G(b,b')$.}
 Since intervals of $G$ are $\tau$-thin, we get $2\delta = d_G(a,c)+d_G(b,d)- (d_G(a,b)+d_G(c,d)) \leq p+e+1+q+t+1+\tau+e+m - (p+e+1+t+m+e+1+q) = \tau$, i.e., \heather{$2\delta\le \tau$}.
 \end{case}

   \begin{figure}[htb]
      \vspace*{-3mm}
     \begin{center} %
        \begin{minipage}[b]{4.3cm}
         \begin{center} \includegraphics[height=3.4cm]{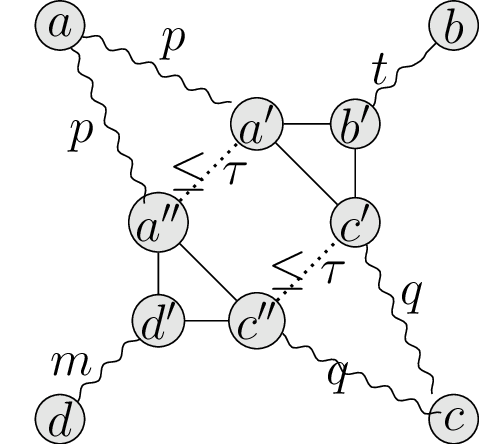}
         \end{center}
         \caption{\label{fig:thinness-to-hyperbolicity-c4} A special subcase of Case \ref{case3}.} %
       \end{minipage}
       \hspace*{0.3cm}
       \begin{minipage}[b]{4.3cm}
         \begin{center} \includegraphics[height=3.4cm]{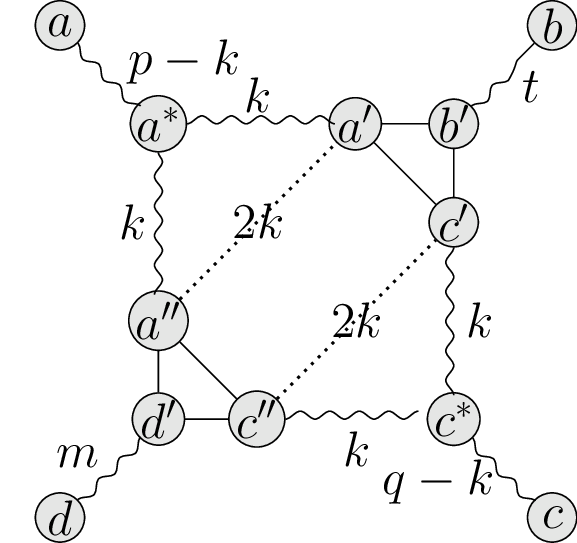}
         \end{center}
         \caption{\label{fig:thinness-to-hyperbolicity-c5} When $\tau(G)=2k$, $\delta(a,b,c,d)=k$.} %
       \end{minipage}
     \end{center} \vspace*{-3mm}
   \end{figure}

 \begin{case}\label{case3}  {\em For vertices $a,b,c$  there are three shortest paths $P(a,b)$, $P(b,c)$, $P_b(a,c)$
  and a triangle $\bigtriangleup (b',c',a')$ in $G$ with edge $a'b'$ on $P(a,b)$, edge $b'c'$ on $P(b,c)$ and
  edge $a'c'$ on $P_b(a,c)$. For vertices $a,d,c$ there are three shortest paths $P(a,d)$, $P(d,c)$,
  $P_d(a,c)$ and a triangle $\bigtriangleup (d',c'',a'')$ in $G$ with edge $a''d'$
  on $P(a,d)$, edge $d'c''$ on $P(d,c)$ and edge $a''c''$ on $P_d(a,c)$.}

 This situation is shown on Fig. \ref{fig:thinness-to-hyperbolicity-c3}.
If vertices $a',a''$ are not in the same slice of $I(a,c)$ then
\heather{set $p := d_G(a,a'')$, and let vertex $c^*$ denote the vertex on $P_b(a,c)$ such that $d_G(a,c^*) = p + 1$.
Set $e := d_G(c^*,a')$ and $m := d_G(d, d')$.
Then, $2\delta = d_G(a,c)+d_G(b,d)- (d_G(a,b)+d_G(c,d)) \le p+1+e+1+q+t+1+e+\tau+1+m -(p+1+e+1+t+q+1+e+1+m) = \tau$.}

 If vertices $a',a''$ are in the same slice of $I(a,c)$ (see Fig.
 \ref{fig:thinness-to-hyperbolicity-c4} for this special subcase; only in this
 subcase \heather{we may have $2\delta = \tau+1$}) then, using notations from Fig.
 \ref{fig:thinness-to-hyperbolicity-c4},
 \heather{
 $2\delta= d_G(a,c)+d_G(b,d)- (d_G(a,b)+d_G(c,d)) \le p+1+q+t+1+\tau+1+m -(p+1+t+m+1+q) =  \tau+1$.
 Furthermore, if $2\delta = \tau+1$}, then $d_G(a',a'')=d_G(c',c'')=\tau$.

 \heather{Assume that $2\delta = \tau+1$} and $\tau$ is even (see Fig. \ref{fig:thinness-to-hyperbolicity-c5}).  Let $\tau=2k$. Consider disks $D(a,p-k)$,$D(a'',k)$, $D(a',k)$ in $G$. These disks pairwise intersect. Hence, there must exist a vertex $a^*$ at distance $p-k$ from $a$ and at distance $k$ from both $a'$ and $a''$. Similarly, there is a vertex $c^*$ in $G$ at distance $q-k$ from $c$ and at distance $k$ from both $c'$ and $c''$. These vertices $a^*$ and $c^*$ belong to slice $S_{t+1+k}(b,d)$ of $I(b,d)$. Hence, $d_G(a^*,c^*)\le \tau=2k$ must hold. On the other hand, $p+1+q=d_G(a,c)\leq d_G(a,a^*)+d_G(a^*,c^*)+d_G(c^*,c)\le p-k+2k+q-k=p+q$, a contradiction. Thus, when $\tau$ is even, \heather{$2\delta=\tau$.}

 \heather{Assume now that $2\delta = \tau+1$} and $\tau$ is odd.  Let $\tau=2k+1$.  As $d_G(a',a'')=2k+1$ and $d_G(a,a')=d_G(a,a'')=p$, by Lemma \ref{centroids}, there must exist three shortest paths $P(a,a')$, $P(a,a'')$, $P(a',a'')$ and a triangle $\bigtriangleup (x,y,z)$ in $G$ with edge $xy$ on $P(a,a')$, edge $xz$ on $P(a,a'')$ and edge $zy$ on
 $P(a',a'')$ (note that $P(a,a')$, $P(a,a'')$, $P(a',a'')$ cannot have a common vertex because of distance requirements). Similarly,
 there must exist three shortest paths $P(c,c')$, $P(c,c'')$, $P(c',c'')$ and a triangle $\bigtriangleup (u,v,w)$ in $G$
 with edge $uv$ on $P(c,c')$, edge $uw$ on $P(c,c'')$ and edge $vw$ on $P(c',c'')$. Thus, by distance requirements,
 four triangles $\bigtriangleup (x,y,z)$, $\bigtriangleup (a',b',c')$, $\bigtriangleup (u,v,w)$, $\bigtriangleup (d',a'',c'')$
 with corresponding shortest paths $P(y,a')\subseteq P(a,a')$, $P(a'',z)\subseteq P(a'',a)$, $P(c'',w)\subseteq P(c'',c)$, $P(c',v)\subseteq P(c',c)$
 of length $k=\lfloor\frac{\tau(G)}{2}\rfloor$ each form in $G$ an $\{x,b',u,d'\}$-distance-preserving subgraph isomorphic to the one depicted on Fig. \ref{fig:unhellified-frames}(c) with $k=l$.

 To complete the proof, it is enough to verify  that if $\tau(G)$ is odd and there exists in $G$ an $\{a,b,c,d\}$-distance-preserving subgraph depicted on Fig. \ref{fig:unhellified-frames}(c) with $k=l=\lfloor\frac{\tau(G)}{2}\rfloor$, then \heather{we obtain $2hb(a,b,c,d)=\tau(G)+1$.} \qed
 \end{case}
 \end{proof}


The following lemmas prove that the three $\{a,b,c,d\}$-distance preserving subgraphs shown in Fig. \ref{fig:unhellified-frames}
can be isometrically embedded into three Helly graphs termed $H_1^{k,l}$, $H_2^{k,l}$, and $H_3^{k,l}$, respectively.
\heather{
Each of $H_1^{k,l}$, $H_2^{k,l}$, and $H_3^{k,l}$ is an isometric subgraph of a King-grid  (Fig. \ref{fig:frames-hellified} gives small examples for $k=l=2$).
Each is induced by the vertices in red as demonstrated in Fig.~\ref{fig:frames-hellified}(a), Fig.~\ref{fig:frames-hellified}(b), and Fig.~\ref{fig:frames-hellified}(c)
such that its four extreme vertices correspond to the four extreme vertices of
an $\{a,b,c,d\}$-distance preserving subgraph shown in Fig.~\ref{fig:unhellified-frames}(a), Fig.~\ref{fig:unhellified-frames}(b), and Fig.~\ref{fig:unhellified-frames}(c), respectively.
In the description that follows let vertices of the form $x_y$ and $x_z$ denote neighbors
which are adjacent to vertex $x$
such that $x_y \in I(x,y)$ and $x_z \in I(x,z)$.
Thus, in $H := H_1^{k,l}$
we have that $d_H(a,d)=d_H(b,c)=k$ and $d_H(a,b)=d_H(d,c)=l$.
In $H := H_2^{k,l}$,
we have that $d_H(a_d,d)=d_H(b,c_b)=k$ and $d_H(a_b,b)=d_H(d,c_d)=l$.
Finally, in $H := H_3^{k,l}$,
we have that $d_H(a_d,d_a)=d_H(b_c,c_b)=k$ and $d_H(a_b,b_a)=d_H(d_c,c_d)=l$.
}

We will show in Section \ref{sec:obstructions} that any Helly graph $G$ with $hb(G)=\delta$
has an isometric $H_1^k$, $H_2^k$, or $H_3^k$, where $k$ is a function of $\delta$. These isometric subgraphs will
be equally important as forbidden subgraphs for $hb(G) \leq \delta$ in Section \ref{sec:obstructions}.
We provide the hellification of all three graphs here for completeness, however, the remainder of this section will
use only the graph in Fig. \ref{fig:unhellified-frames}(c) and its hellification $H_3^{k,l}$ in order to
refine result of Lemma \ref{hyperbolicity-less-than-thiness} in the special case \heather{when $2hb(G)=\tau(G)+1$.}

\begin{figure}
    [hbt] 
     \vspace*{-.2cm}
    \centering
    \image{$H_1^{k,l}$}{
       \includegraphics[scale=.43]{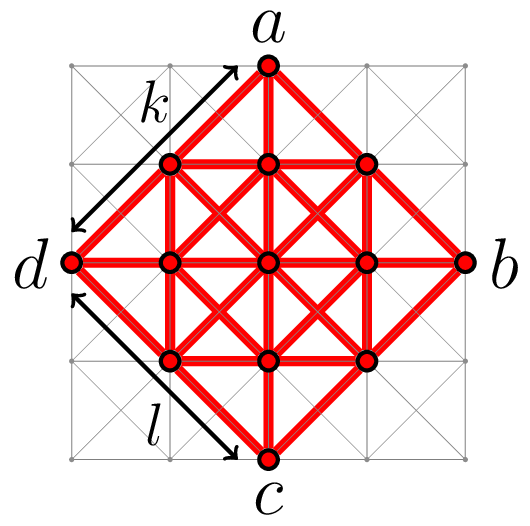}
     }\quad
     \image{$H_2^{k,l}$}{
       \includegraphics[scale=.43]{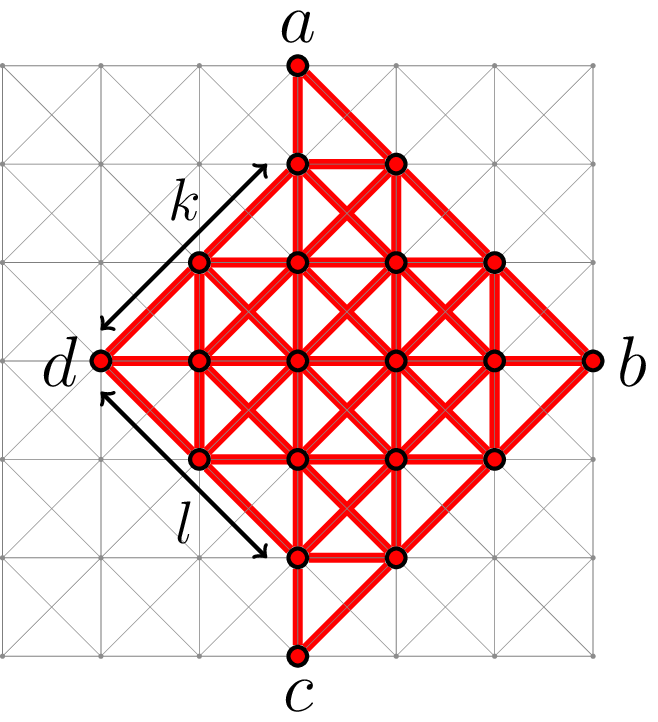}
     }\quad
     \image{$H_3^{k,l}$}{
       \includegraphics[scale=.43]{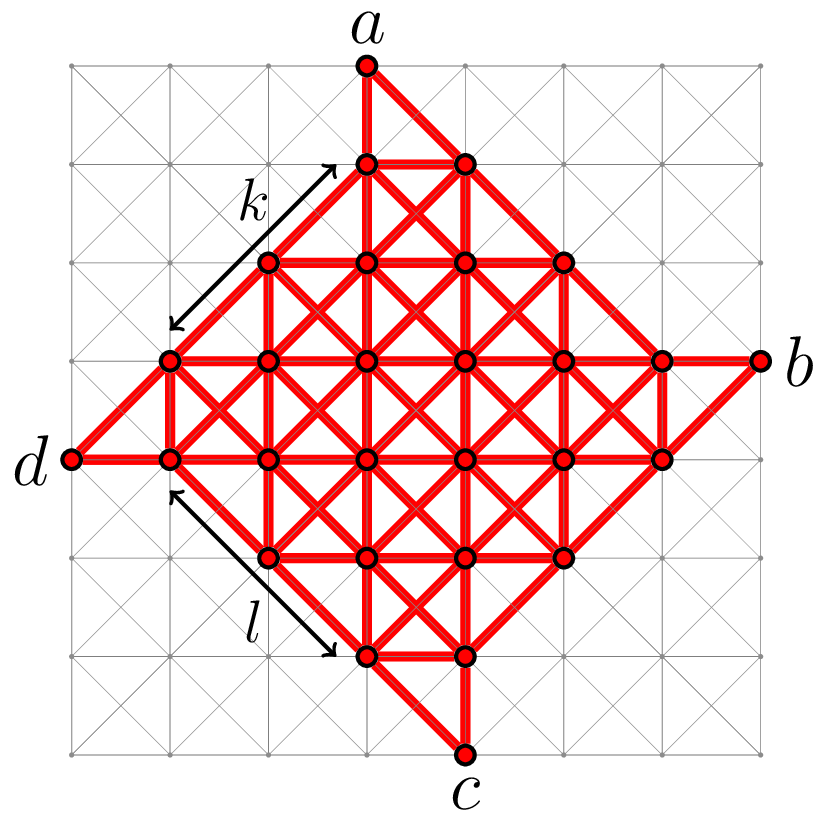}
     }
    \caption
    {
        Examples of $H_1^{k}$, $H_2^{k}$, and $H_3^{k}$ \heather{shown in red,} where $k=l=2$, based
        on respective inputs from Fig. \ref{fig:unhellified-frames}. We omit the second superscript
        and use the notation $H_i^k$ when $k=l$.
        Isometric embeddings of those graphs into the King-grid are shown.
    }
    \label{fig:frames-hellified}
\end{figure}
\setcounter{subfloat}{0}
%


 \begin{restatable}{lemma}{sframehellified}\label{s-frame-hellified}
   If a Helly graph $G$ has an $\{a,b,c,d\}$-distance preserving subgraph
   depicted on Fig. \ref{fig:unhellified-frames}(a), then $G$ has an isometric subgraph $H_1^{k,l}$
   with $a,b,c,d$ as corner points (see Fig. \ref{fig:frames-hellified}(a)).
 \end{restatable}
 \begin{proof} Let $a,b,c,d$ be four vertices of a Helly graph $G$  such that
 $d_G(a,c)=$ $d_G(b,d)=k+l$ and $d_G(a,b)=$ $d_G(c,d)=$ $l$ and $d_G(b,c)=$ $d_G(d,a)=$ $k$. Let $P(d,a)=(d,v_1,v_2,\dots, v_k=a)$,
 $P(d,c)=(d,u_1,u_2,\dots, u_l=c)$ be two shortest paths connecting the appropriate vertices. Consider disks $D(v_1,1)$, $D(u_1,1)$, $D(b,k+l-2)$ in $G$.
These disks pairwise intersect. Hence, by the Helly property, there is a vertex $d'$ which is adjacent to both $v_1$ and $u_1$ and at distance \guarnera{$k+l-2$} from $b$. Since disks $D(a,1), D(b,l-1), D(d',k-1)$ pairwise intersect, there must exist a vertex $a'$ such that $a'$ is adjacent to $a$ and at distance $k-1$ from $d'$ and distance $l-1$ from $b$. Similarly, considering pairwise intersecting disks  $D(c,1), D(b,k-1), D(d',l-1)$, there exists a vertex  $c'$  which is
adjacent to $c$ and at distance $l-1$ from $d'$ and distance $k-1$ from $b$.
 For vertices $a',b,c',d'$ we have $d_G(a',c')=d_G(b,d')=l+k-2$ and $d_G(a',b)=d_G(c',d')=l-1$ and $d_G(b,c')=d_G(d',a')=k-1$. Hence, by induction, we may assume that in $G$ there is an isometric subgraph $H_1^{k-1,l-1}$ with $a',b,c',d'$ as corner points. In what follows, using the Helly property, we extend this $H_1^{k-1,l-1}$ to isometric $H_1^{k,l}$ with $a,b,c,d$ as corner points (see Fig. \ref{fig:proof-sframes} for an illustration).

 Let $P(d',a')=(d'=v'_1,v'_2,\dots, v'_k=a')$ be the shortest path of $H_1^{k-1,l-1}$ connecting $d'$ with $a'$. For each edge $v'_iv'_{i+1}$ of this path, denote by $w_i$ a vertex of $H_1^{k-1,l-1}$ which forms a triangle with $v'_iv'_{i+1}$.
 \guarnera{
Let~$P(d,a)$ denote path $(d=v_0,v_1,v_2,\dots, v_k=a)$.
First, we show that each vertex $v_i \in P(d,a)$ for $i=1,2,...,k-1$ can be chosen such that~$v_iv'_i$ is an edge of~$G$ for each~$i$.
}
Let \guarnera{$i\ge 1$} be the smallest index such that $v_iv'_i\notin E$. Consider pairwise intersecting disks $D(v_{i-1},1), D(v'_{i},1), D(a,d_G(a,v_i))$. By the Helly property, there is a vertex $v_i^*$ in $G$ which is adjacent to both $v_{i-1}$ and $v'_{i}$ and at distance $d_G(a,v_i)$ from $a$. Hence, we can replace part of $P(d,a)$ from $v_i$ to $a$ with a new
shortest path from $v_i^*$ to $a$. So, we can assume that  $v_iv'_i\in E$ for each $i$. Since vertices $a,v'_k,w_{k-1},v'_{k-1},v_{k-1}$ are pairwise at distance at most 2, by the Helly property, there must exist a vertex $w'_{k-1}$ which is adjacent to all $a,v'_k,w_{k-1},v'_{k-1},v_{k-1}$. Having vertex $w'_{k-1}$, we can use the Helly property to impose a new vertex $w'_{k-2}$ adjacent to all $v_{k-1},v'_{k-1},w'_{k-1},w_{k-2},v'_{k-2},v_{k-2}$. Continuing this way, we obtain a new vertex
$w'_{i}$ (for $i=k-3, k-4,\dots,1$) which is adjacent to all $v_{i+1},v'_{i+1},w'_{i+1},w_{i},v'_{i},v_{i}$.
 This completes the addition to $H_1^{k-1,l-1}$ along the path $P(d,a)=(d,v_1,v_2,\dots, v_k=a)$. Similarly,
 the addition along the path $P(d,c)=(d,u_1,u_2,\dots, u_l=c)$ can be done completing the extension of $H_1^{k-1,l-1}$ to $H_1^{k,l}$ which is clearly an isometric subgraph of $G$.  \qed

\begin{figure}
    [hbt] 
     \vspace*{-.2cm}
    \centering
    \image{}{
  \includegraphics[scale=.69]{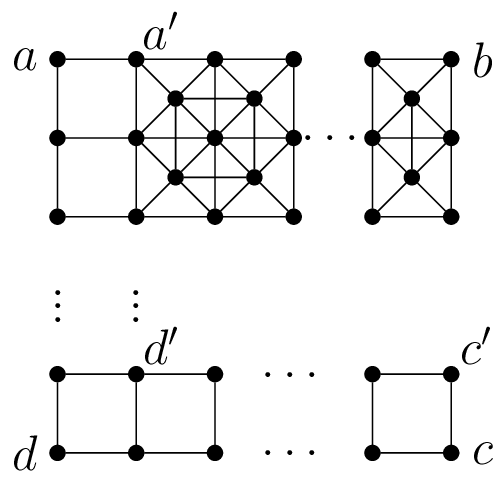}
 }\quad
 \image{}{
  \includegraphics[scale=.69]{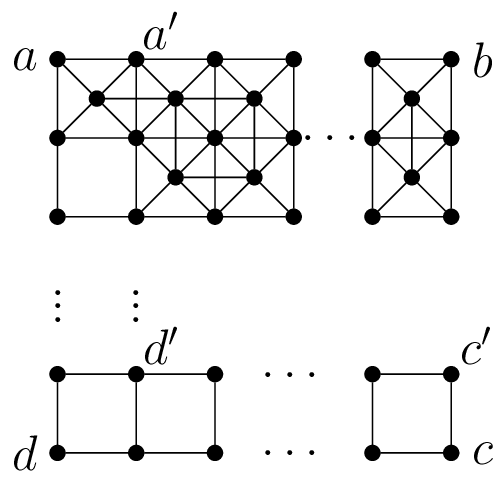}
 }\quad
 \image{}{
  \includegraphics[scale=.69]{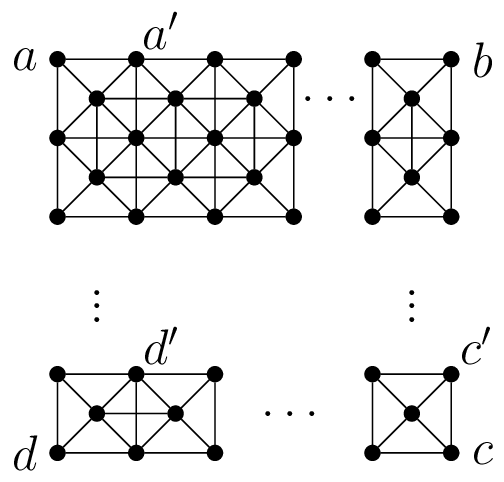}
 }
    \caption
    {
        Extension of an isometric subgraph $H_1^{k-1,l-1}$ to an isometric subgraph $H_1^{k,l}$.
        Isometric embeddings of those graphs into the King-grid are shown.
    }
    \label{fig:proof-sframes} %
\end{figure}
%
%
 \end{proof}

 \begin{restatable}{lemma}{dframehellified}\label{d-frame-hellified}
   If a Helly graph $G$ has an $\{a,b,c,d\}$-distance preserving subgraph
   depicted on Fig. \ref{fig:unhellified-frames}(b), then $G$ has an isometric subgraph $H_2^{k,l}$
   with $a,b,c,d$ as corner points (see Fig. \ref{fig:frames-hellified}(b)).
 \end{restatable}
 \begin{proof}
   Let $a$, $b$, $c$, $d$ be vertices of a Helly graph $G$, with $\bigtriangleup(a,a_b,a_d)$
   and $\bigtriangleup(c,c_b,c_d)$ such that
   $d_G(a,c)=k+l+$ $2 = d_G(a,a_b)+d_G(a_b,c_b)+d_G(c_b,c) = d_G(a,a_d)+d_G(a_d,c_d)+d_G(c_d,c)$,
   $d_G(b,d)=k+l+1 = d_G(d,a_d)+1+d_G(a_b,b) = d_G(d,c_d)+1+d_G(c_b,b)$, and
   $d_G(b,c)=d_G(d,a)=k+1$ and $d_G(a,b)=d_G(c,d)=l+1$.
   Consider disks $D(a_b,k)$, $D(c_b,l)$, and $D(d,1)$ in $G$. These disks pairwise intersect.
   Hence, by the Helly property, there is a vertex $d'$ which is adjacent to $d$
   and at distance $k$ from $a_b$ and at distance $l$ from $c_b$. For vertices $a_b,b,c_b,d'$,
   we have $d_G(a_b,b)=d_G(c_b,d')=l$, $d_G(b,c_b)=d_G(d',a_b)=k$, and
   $d_G(a_b,c_b)=d_G(d',b)=k+l$. By Lemma \ref{s-frame-hellified}, there is
   an isometric subgraph $H_1^{k,l}$ with $a_b,b,c_b,d'$ as corner points (see
   Fig. \ref{fig:proof-dframes}).

   Let $P(a_b,d') = (a_b = v'_0, v'_1, v'_2, ..., v'_k = d')$ be the shortest path of
   $H_1^{k,l}$ connecting $d'$ with $a_b$. For each edge $v'_{i}v'_{i+1}$ of this
   path, denote by $w'_{i+1}$ a vertex of $H_1^{k,l}$ which forms a triangle with $v'_{i}v'_{i+1}$.
   Since vertices $d,v'_k,v'_{k-1},w'_k$ are pairwise distant at most 2 and
   distant from $a$ at most $k+1$, by the Helly property there must exist a
   vertex $v^*_{k}$ adjacent to $d,v'_k,v'_{k-1},w'_k$ and at distance $k$ from $a$.
   Having vertex $v^*_{k}$, we can use the Helly property to impose a new vertex
   $v^*_{k-1}$ which is adjacent to all $v^*_{k},v'_{k-1},v'_{k-2},w'_{k-1}$
   and at distance $k-1$ from $a$. Continuing this way, we obtain a new vertex
   $v^*_i$ which is adjacent to all $v^*_{i+1},v'_{i},v'_{i-1},w'_{i}$ and at
   distance $i$ from $a$ (for $i=k-2,k-3,...,1$). This completes the addition
   to $H_1^{k,l}$ along the path $P(a_b,d')$.
   Similarly, the addition along the path $P(c_b,d') = (c_b = u'_0, u'_1, u'_2, ..., u'_l = d')$
   can be done.
   This completes the extension of $H_1^{k,l}$ to $H_2^{k,l}$.

   Clearly, $H_2^{k,l}$ obtained from $H_1^{k,l}$ is an isometric subgraph of $G$. Recall that $H_2^{k,l}$ is a $\{a,b,c,d\}$-distance preserving subgraph of $G$. We know from Lemma \ref{s-frame-hellified} that
   $H_1^{k,l}$-part of $H_2^{k,l}$ is an isometric subgraph of $G$. We know also that every pair $x,y \in H_2^{k,l} \setminus H_1^{k,l}$
   belongs to a shortest path of $G$ from $a$ to $c$ passing through $d$. Finally, every pair $x,y$ with  $x \in H_2^{k,l} \setminus H_1^{k,l}$ and  $y \in H_1^{k,l}$ belongs to a shortest path of $G$ connecting $a$ with $c$ or $b$ with $d$.
   \qed
    \begin{figure}[!hbt]
    \begin{centering}
      \includegraphics[scale=.6]{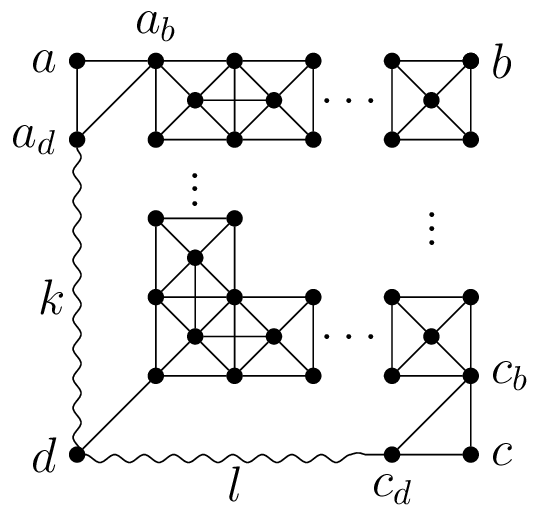}
      \caption{Using the Helly property, the graph from Fig. \ref{fig:unhellified-frames}(b) is shown
   to have $H_2^{k,l}$ as an isometric subgraph.}\label{fig:proof-dframes}
    \end{centering}
    \end{figure}
  \end{proof}

  \begin{restatable}{lemma}{fframehellified}\label{f-frame-hellified}
   If a Helly graph $G$ has an $\{a,b,c,d\}$-distance preserving subgraph
   depicted on Fig. \ref{fig:unhellified-frames}(c), then $G$ has an isometric subgraph $H_3^{k,l}$
   with $a,b,c,d$ as corner points (see Fig. \ref{fig:frames-hellified}(c)).
 \end{restatable}
 \begin{proof}
   Let $a,b,c,d$ be vertices of a Helly graph $G$ such that $d_G(a,c)=d_G(b,d)=l+k+3$
   and $d_G(a,b)=d_G(c,d)=l+2$ and $d_G(b,c)=d_G(d,a)=k+2$ (see Fig. \ref{fig:proof-fframes}(a)).
   Since $d(a_b,c_b)=k+l+1$, $d(a_b,d)=k+1+1$ and $d(c_b,d)=1+l+1$, by Lemma \ref{centroids},
   there is a triangle $\bigtriangleup(d',d'_a,d'_c)$ such that
   $d$ is adjacent to $d'$ and $d_G(d'_c,c_b) = l $, $d_G(d'_a,a_b) = k$.
 For vertices $a_b, b,c_b,d'$, we have $d_G(a_b,b) = d_G(c_b,d') = l+1$, $d_G(b,c_b) =  d_G(d',a_b) = k+1$,
   and $d_G(d',b) = k+l + 2$, as well as $d_G(a_b,c_b) = k+l+1$. By Lemma
   \ref{d-frame-hellified}, there is in $G$ an isometric $H_2^{k,l}$ with $a_b, b,c_b,d'$
   as corner points (see Fig. \ref{fig:proof-fframes}(b)).

   Let $P(a_b,d') = (a_b = v'_0, v'_1, v'_2, ..., v'_k, d')$ be the shortest
   path of $H_2^{k,l}$ connecting $d'$ with $a_b$, and
   let $P(c_b,d') = (c_b = u'_0, u'_1, u'_2, ..., u'_l, d')$ be the shortest
   path of $H_2^{k,l}$ connecting $d'$ with $c_b$. For each edge $v'_iv'_{i+1}$,
   denote by $w_{i+1}$ a vertex of $H_2^{k,l}$ which forms a triangle with $v'_iv'_{i+1}$.
   Since vertices $d',d,v'_k,u'_l$ are pairwise at distance at most 2 and at distance at most $k+2$ from
   $a$, by the Helly property, there must exist a vertex $v^*_{k+1}$
   adjacent to $d',d,v'_k,u'_l$ and at distance $k+1$ from $a$.
   Having vertex $v^*_{k+1}$, we can use the Helly property to impose a new vertex
   $v^*_{k}$ which is adjacent to all $v^*_{k+1},w_{k},v'_{k},v'_{k-1}$ and at distance
   $k$ from $a$. Continuing this way, we obtain a new vertex $v^*_i$ which
   is adjacent to all $v^*_{i+1},w_{i},v'_{i},v'_{i-1}$ and at distance $i$ from $a$
   (for $i=k-1, k-2, ..., 1$).

   For each edge $u'_iu'_{i+1}$
   denote by $y_{i+1}$ a vertex of $H_2^{k,l}$ which forms a triangle with $u'_iu'_{i+1}$.
   Since vertices $d,v^*_{k+1},u'_l,v'_k$ are pairwise at distance at most 2 and at distance  at most $l+2$ from
   $c$, by the Helly property, there must exist a vertex $u^*_{l+1}$
   adjacent to $d,v^*_{k+1},u'_l,v'_k$ and at distance $l+1$ from $c$.
   Having vertex $u^*_{l+1}$, we can use the Helly property to impose a new
   vertex $u^*_{l}$ which is adjacent to all $u^*_{l+1},u'_{l},u'_{l-1},y_{l}$
   and at distance $l$ from $c$.
   Continuing this way, we obtain a new
   vertex $u^*_i$ which is adjacent to all $u^*_{i+1},u'_{i},u'_{i-1},y_{i}$
   and is at distance $i$ from $c$ (for $i=l-1,l-2,...,1)$. This completes
   the extension of $H_2^{k,l}$ to $H_3^{k,l}$.

     Clearly, $H_3^{k,l}$ obtained from $H_2^{k,l}$ is an isometric subgraph of $G$. Recall that $H_3^{k,l}$ is a $\{a,b,c,d\}$-distance preserving subgraph of $G$. We know from Lemma \ref{d-frame-hellified} that
   $H_2^{k,l}$-part of $H_3^{k,l}$ is an isometric subgraph of $G$. We know also that every pair $x,y \in H_3^{k,l} \setminus H_2^{k,l}$
   belongs to a shortest path of $G$ from $a$ to $d$ or from $d$ to $c$ or from $a$ to $c$ passing through a neighbor of $d$. Finally, every pair $x,y$ with  $x \in H_3^{k,l} \setminus H_2^{k,l}$ and  $y \in H_2^{k,l}$ belongs to a shortest path of $G$ connecting $s$ with $t$ where $s,t\in \{a,b,c,d\}$.
   \qed

   \begin{center}
   \centering\image{}{
     \includegraphics[scale=.69]{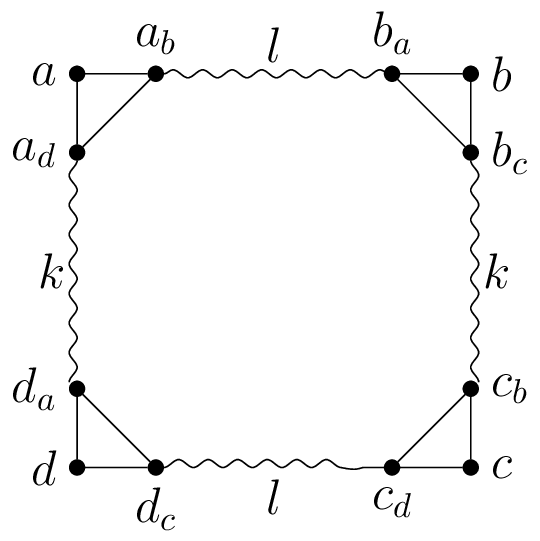}
   }\quad
   \image{}{
     \includegraphics[scale=.69]{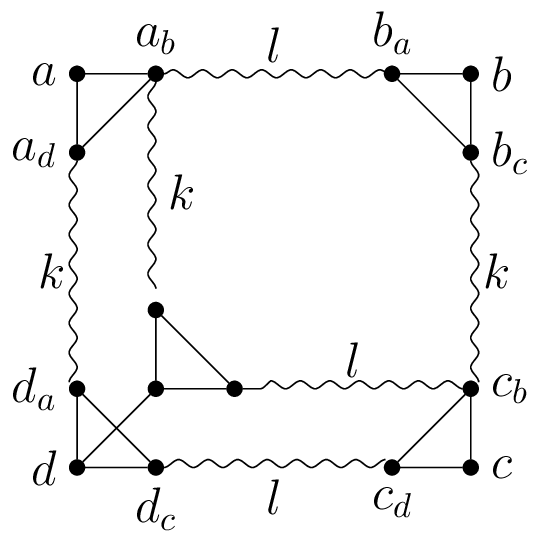}
   }
   \captionof{figure}{Using the Helly property, the graph from
   Fig. \ref{fig:unhellified-frames}(c) is shown
   to have $H_3^{k,l}$ as an isometric subgraph.}
   \label{fig:proof-fframes}
   \end{center}
   \setcounter{subfloat}{0}
 \end{proof}

Combining Lemma \ref{half-thin}, Lemma \ref{hyperbolicity-less-than-thiness}  and Lemma \ref{f-frame-hellified}, we conclude
with a tight bound on hyperbolicity with respect to interval thinness in Helly graphs, as well as with a characterization
of the case in which the hyperbolicity of a Helly graph realizes the upper bound.

 \begin{theorem} \label{thinness-hyperbolicity}
  For every Helly graph $G$, \heather{$\tau(G) \leq 2hb(G) \leq \tau(G)+1$. Furthermore, $2hb(G)=\tau(G)+1$}
   if and only if $\tau(G)$ is odd and 
 $G$ contains graph $H_3^{k}$ with $k=\lfloor\frac{\tau(G)}{2}\rfloor$ as an isometric subgraph.
 \end{theorem}

 \begin{corollary} \label{cor:thinness-hyperbolicity}
  For every Helly graph $G$, if $\tau(G)$ is even, then $hb(G)$ is an integer and \heather{$2hb(G)=\tau(G)$.}
 \end{corollary}

 %
 %
 \section{Three isometric subgraphs of the King-grid are the only obstructions to a small hyperbolicity in Helly graphs} \label{sec:obstructions}

In this section, we will identify three isometric subgraphs of the King-grid that are responsible for the hyperbolicity of a Helly graph $G$.
These are named $H_1^k$, $H_2^k$, $H_3^k$, and are shown in Fig. \ref{fig:frames-hellified}. We may assume that $hb(G)>0$ as the structure of any graph with hyperbolicity 0 is well-known; they are exactly the block graphs, i.e., graphs where each biconnected component is a complete graph~\cite{Gromov87}.


The following lemma shows the existence of one of the three isometric subgraphs in a Helly graph $G$ with $hb(G)=k>0$.
\begin{lemma} \label{lem:new-helly-symmetric-reduction}
  Let $G$ be a Helly graph with $hb(G) = k > 0$. \\
If $\tau(G)=2k$ and $k$ is an integer, then $G$ contains $H_1^k$  as an isometric subgraph.\\
If $\tau(G)=2k$ and $k$ is a half-integer, then $G$ contains $H_2^{k - \frac{1}{2}}$ as an isometric subgraph. \\
If $\tau(G)=2k-1$, then $k$ is an integer and $G$ contains $H_3^{k-1}$  as an isometric subgraph.
\end{lemma}
\begin{proof}
  Let $hb(G) = k > 0$, and let interval $I(x,y)$
  realize the maximum thinness, that is there are vertices $z,t \in S_\alpha(x,y)$, for some integer $\alpha$,  such that $d(z,t)=\tau(G)$.
  By Theorem \ref{thinness-hyperbolicity}, either $\tau(G)=2k$ or $\tau(G)=2k-1$. 
  If $\tau(G)=2k-1$, then by Theorem \ref{thinness-hyperbolicity}, $\tau(G)$ is odd (thus $k$ is an integer) and $G$ contains $H_3^{k-1}$
  as an isometric subgraph.
  If $\tau(G)=2k$, then $\tau(G)$ can be even or odd (since $k$ can be a half-integer).
  Set $\alpha := d(x,t) = d(x,z)$, and $\beta := d(t,y) = d(z,y)$.

  Let $\tau(G)=2k$ be even (thus $k$ is an integer).
  Clearly $\alpha \geq k$ and $\beta \geq k$, otherwise $d(z,t) < 2k$.
  By Lemma \ref{fnew},
  there is a vertex $x'$ such that $d(x,x')=\alpha-k$, $d(z,x')=k$, and $d(t,x')=k$, and
  there is a vertex $y'$ such that $d(y,y')=\beta-k$, $d(z,y')=k$, and $d(z,y')=k$.
  By the triangle inequality, $d(x',y') \leq d(x',z) + d(z,y') = 2k$ and $\alpha + \beta = d(x,y) \leq \alpha - k + d(x',y') + \beta - k\le \alpha + \beta$.
  Therefore, $d(x',y')=2k$ must hold. Then, by Lemma \ref{s-frame-hellified}, $G$ contains an isometric subgraph $H_1^k$ with $\{x',z,y',t\}$ as corner points.

  Let $\tau(G)=2k$ be odd (thus $k$ is a half-integer).
  Let $k=p + \frac{1}{2}$ for an integer $p$. Then $d(z,t)=2p+1$.
  Clearly $\alpha > p$ and $\beta > p$, otherwise $d(z,t) < 2p+1$.
  By Lemma \ref{fnew}, there is a triangle $\triangle(x',x_z,x_t)$ such that
  $d(x,x')=\alpha-p-1$, $d(x_z,z)=p$, and $d(x_t,t)=p$,
  and there is a triangle $\triangle(y',y_z,y_t)$ such that
  $d(y,y')=\beta-p-1$, $d(y_z,z)=p$, and $d(y_t,t)=p$.
  By the triangle inequality, $d(x',y') \leq d(x',x_z) + d(x_z,z)+d(z,y_z,)+d(y_z,y')= 2p + 2$ and
  $\alpha + \beta = d(x,y) \leq \alpha - p - 1 + d(x',y') + \beta - p - 1 = \alpha + \beta$.
  Therefore, $d(x',y')=2p+2$.
  Since $p=k-\frac{1}{2}$, by Lemma \ref{d-frame-hellified}, $G$ contains an isometric subgraph $H_2^{k - \frac{1}{2}}$ with $\{x',z,y',t\}$ as corner points.~\qed
\end{proof}

Using the previous lemma, we can now characterize  Helly graphs $G$ with $hb(G) \leq \delta$
based on three forbidden isometric subgraphs. Whether $\delta$ is an integer or a half-integer
determines which of the $H_1^k$, $H_2^k$, $H_3^k$ graphs are forbidden and the
value of $k$.

\begin{theorem} \label{general-forbidden}
  Let $G$ be a Helly graph and $k$ be a non-negative integer. \\
  - $hb(G) \leq k$
  if and only if
  $G$ contains no $H_2^{k}$ as an isometric subgraph. \\
 - $hb(G) \leq k+\frac{1}{2}$
  if and only if
  $G$ contains neither $H_1^{k+1}$ nor $H_3^{k}$ as an isometric subgraph.
\end{theorem}

\begin{proof}
  Assume $hb(G) \leq k$ and that
  $G$ has $H_2^{k}$ as an isometric subgraph.
  It is easy to check that  $hb(H_2^{k}) = k + \frac{1}{2} > k$ (the hyperbolicity realizes on four extreme vertices).
  As the hyperbolicity of a graph is at least the hyperbolicity of its
  isometric subgraph, $hb(G) > k$, giving a contradiction.

  Assume $hb(G) \leq k+\frac{1}{2}$, and that
  $G$ has $H_1^{k+1}$ or $H_3^{k}$  as an isometric subgraph.
  It is easy to check that $hb(H_1^{k+1}) = k+1 > k+\frac{1}{2}$ and
  $hb(H_3^{k}) = k+1 > k+\frac{1}{2}$ (the hyperbolicity of each realizes on four extreme vertices).
  As the hyperbolicity of a graph is at least the hyperbolicity of its
  isometric subgraph, $hb(G) > k+\frac{1}{2}$, giving a contradiction.

  For the other direction, assume $hb(G)=\delta$. Then, by Lemma \ref{lem:new-helly-symmetric-reduction}, $G$ has one of $H_1^\delta$, $H_2^{\delta-\frac{1}{2}}$, $H_3^{\delta-1}$ as an isometric subgraph. Note that, for any integer $m$, $H_1^m$ is an isometric subgraph of $H_2^m$ and $H_1^{m+1}$, $H_2^m$ is an isometric subgraph of $H_3^{m}$, $H_2^{m+1}$ 
  and $H_1^{m+1}$, and $H_3^m$ is an isometric subgraph of $H_3^{m+1}$. If $\delta$ is an integer, $G$ contains $H_1^\delta$ or $H_3^{\delta-1}$, and hence $H_1^{k+1}$ or $H_3^{k}$ when $\delta>k+\frac{1}{2}$, as an isometric subgraph.
  If $\delta$ is a half-integer, $G$ contains $H_2^{\delta-\frac{1}{2}}$, and hence $H_2^{k}$ when $\delta\geq k+\frac{1}{2}$, as an isometric subgraph. \qed
\end{proof}

Theorem  \ref{general-forbidden}
can easily be applied to determine the forbidden subgraphs characterizing any
$\delta$-hyperbolic Helly graph. The corollaries that follow exemplify
this for $\frac{3}{2}$-hyperbolic Helly graphs and $2$-hyperbolic Helly graphs.

   \begin{figure}[htb]
  \begin{center}
     \vspace*{-3mm}
     \begin{minipage}[b]{6.7cm}
    \begin{center}
       \centering\image{}{
         \includegraphics[scale=.45]{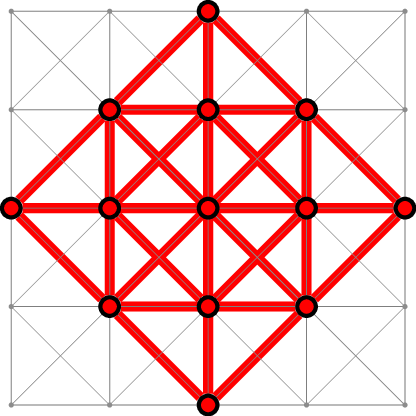}
       }\quad\image{}{
         \includegraphics[scale=.45]{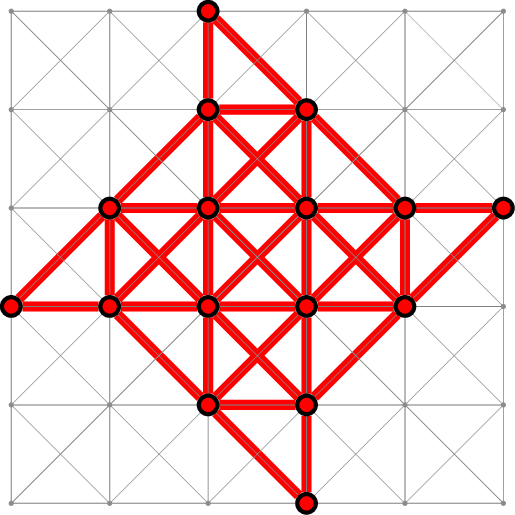}
       }\captionof{figure}{Forbidden isometric subgraphs for $\frac{3}{2}$-hyperbolic Helly graphs.}\label{forbidden-one-and-half-hyperbolic}
       \end{center}\setcounter{subfloat}{0}
       \end{minipage}
      \hspace*{1cm}
      \begin{minipage}[b]{6.5cm}
        \begin{center} \includegraphics[scale=.45]{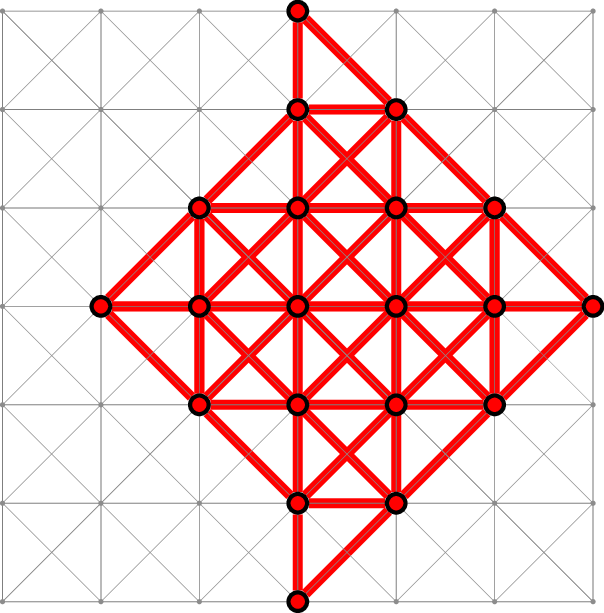}
        \end{center}
        \caption{\label{forbidden-two-hyperbolic}Forbidden isometric subgraph for $2$-hyperbolic Helly graphs.} %
      \end{minipage}
      \vspace*{-3mm}
      \end{center}
  \end{figure}

 \begin{corollary}
   A Helly graph is $\frac{3}{2}$-hyperbolic if and only if
   it contains neither of graphs from Fig. \ref{forbidden-one-and-half-hyperbolic} as an isometric subgraph.
\end{corollary}


 \begin{corollary}
   A Helly graph is $2$-hyperbolic if and only if
    it does not contain graph from Fig. \ref{forbidden-two-hyperbolic} as an isometric subgraph.
\end{corollary}




 To give a few equivalent characterizations of $\frac{1}{2}$-hyperbolic Helly graphs, we will need one more lemma.
 Let $C_4$ denote an induced cycle on four vertices. We say that a graph $G$ is {\em $C_4$-free} if it does not contain $C_4$ as an induced subgraph.
 The graph $H_3^0$ is also known in the literature as the 4-sun $S_4$ \heather{(see Fig. \ref{fig:s4}).}

 \begin{lemma} [\cite{Dragan1993}] \label{c4-g2-then-s4}
   For any $C_4$-free Helly graph $G$, every $C_4$ in $G^2$ forms in $G$ an isometric subgraph $S_4$.
 \end{lemma}

\heather{
   \begin{figure}[bt]
      \centering\begin{minipage}[b]{\linewidth}
       \centering\includegraphics[scale=.49]{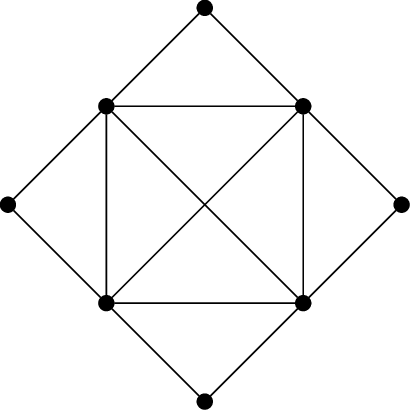}
       \caption{\label{fig:s4}The graph $H_3^0$, also known as the 4-sun $S_4$.} %
     \end{minipage}
 \end{figure}
}

 \heather{
 By combining Theorem \ref{thinness-hyperbolicity}, Theorem \ref{general-forbidden} and Lemma \ref{c4-g2-then-s4},
 we obtain the following characterization of $\frac{1}{2}$-hyperbolic Helly graphs.
 One necessary and sufficient condition is that $G$ and $G^2$ are $C_4$-free.
 In fact, in the characterization of any $\frac{1}{2}$-hyperbolic graph \cite{Coudert14},
 there is a similar requirement that
 every graph power $G^i$ for $i \geq 1$ is $C_4$-free and one additional graph is $C_4$-free.
 We explore this relationship between $C_4$-free graph powers and the $\delta$-hyperbolicity
 of any Helly graph in subsequent results presented here.

 }

 \begin{corollary} \label{half-hyperbolic-characterization}
   The following statements are equivalent for any Helly graph $G$:
 \begin{enumerate}[i)]
 \item $G$ is $\frac{1}{2}$-hyperbolic;
 \item $G$ has neither $C_4$ nor $S_4$ as an isometric subgraph;
 \item Neither $G$ nor $G^2$ has an induced $C_4$;
 \item $\tau(G) \leq 1$ and $G$ has no $S_4$ as an isometric subgraph.
 \end{enumerate} \end{corollary}

  The following lemmas describe the three forbidden isometric subgraphs in terms of graph powers.

  \begin{lemma} \label{lem:h1-power}
    Let $G$ be a Helly graph and $k$ be a non-negative integer. Then $G$ has $H_1^{k+1}$ as an isometric subgraph
    if and only if there exist four vertices in $G$ that form $C_4$ in $G^\ell$ for all $\ell \in [k+1, 2k+1]$.
  \end{lemma}
  \begin{proof}
    Suppose $G$ has $H_1^{k+1}$ as an isometric subgraph.
    Then, for four extreme vertices $x$, $y$, $z$, $t$ of $H_1^{k+1}$, we have $d(x,y)=d(y,z)=d(z,t)=d(t,x)=k+1$  and
    $d(x,z)=d(y,t)=2k+2$. Thus, $x$, $y$, $z$, $t$, form $C_4$ in $G^\ell$ for all $\ell \in [k+1, 2k+1]$.

    Now, let $x$, $y$, $z$, $t$ be four vertices in $G$ that form $C_4$ in $G^\ell$ for all $\ell \in [k+1, 2k+1]$.
    Then, each of $d(x,y)$, $d(y,z)$, $d(z,t)$, $d(t,x)$ is less than or equal to $k+1$, and $d(x,z)$, $d(y,t)$ are
    greater than or equal to $2k+2$. From these distance requirements, necessarily, $d(x,y)=d(y,z)=d(z,t)=d(t,x)=k+1$ and $d(x,z)=d(y,t)=2k+2$. By Lemma \ref{s-frame-hellified}, $G$ has isometric $H_1^{k+1}$. \qed

  \end{proof}

A cycle on 4 vertices with one diagonal is called a {\sl diamond}.

  \begin{lemma} \label{int-graph-power}
    Let $G$ be a Helly graph and $k$ be a non-negative integer. Then $G$ has $H_2^{k}$ as an isometric subgraph if and only if
    there exist four vertices in $G$ that form $C_4$ in $G^\ell$ for all $\ell \in [k+1, 2k]$ and form a diamond in $G^{2k+1}$.
  \end{lemma}

  \begin{proof}
    Suppose $G$ has $H_2^{k}$ as an isometric subgraph. Then, for four extreme vertices $x$, $y$, $z$, $t$ of $H_2^{k}$,
    we have $d(x,y)=d(y,z)=d(z,t)=d(t,x)=k+1$ and $d(x,z)=2k+2$ and $d(y,t)=2k+1$.
    Thus, $x$, $y$, $z$, $t$, form $C_4$ in $G^\ell$ for all $\ell \in [k+1,2k]$ and form a diamond in $G^{2k+1}$.

    Next, let $x$, $y$, $z$, $t$ be four vertices in $G$ that form $C_4$ in $G^\ell$ for all $\ell \in [k+1, 2k]$ and form a diamond in $G^{2k+1}$.
    Then, each of $d(x,y)$, $d(y,z)$, $d(z,t)$, $d(t,x)$ is less than or equal to $k + 1$.
    Without loss of generality, let $yt$ be the chord of a diamond in $G^{2k +1}$ formed by $x$, $y$, $z$, $t$.  Thus, $d(y,t)\geq 2k+1$ and $d(x,z)\geq 2k+2$.
    From these distance requirements, necessarily, $d(x,y)=d(y,z)=d(z,t)=d(t,x)=k+1$, $2k+1\leq d(y,t)\leq 2k+2$ and $d(x,z)= 2k+2$.
    If $d(y,t)= 2k+2$, then by Lemma \ref{s-frame-hellified}, $G$ has an isometric $H_1^{k+1}$, and hence an isometric $H_2^k$ (note that $H_1^{k+1}$ contains an isometric $H_2^k$).
    Let now $d(y,t)= 2k+1$. By Lemma \ref{fnew}, there exist shortest paths $P(x,y)$ and $P(x,t)$ such that the neighbors of $x$ on those paths are adjacent.
    Similarly, there exist shortest paths $P(z,y)$ and $P(z,t)$ such that the neighbors of $z$ on those paths are adjacent. Thus,
    $x$, $y$, $z$, $t$ form $\{x,y,z,t\}$-distance preserving subgraph depicted on Fig. \ref{fig:unhellified-frames}(b).
     By Lemma \ref{d-frame-hellified}, $G$ has an isometric $H_2^k$. \qed
  \end{proof}

The following result generalizes Lemma \ref{c4-g2-then-s4}. 

  \begin{lemma} \label{half-int-graph-power}
    Let $G$ be a Helly graph and $k$ be a non-negative integer. Then $G$ has $H_1^{k+1}$ or $H_3^{k}$ as an isometric subgraph if and only if
    there exist four vertices in $G$ that form $C_4$ in $G^\ell$ for all $\ell \in [k+1, 2k+1]$ or
    there exist four vertices in $G$ that form $C_4$ in $G^\ell$ for all $\ell \in [k+2, 2k+2]$.
  \end{lemma}
  \begin{proof}
    By Lemma \ref{lem:h1-power}, $G$ has $H_1^{k+1}$ as an isometric subgraph if and only if there exist four vertices in $G$ that form $C_4$ in $G^\ell$ for all $\ell \in [k+1, 2k+1]$.
    Suppose $G$ has $H_3^{k}$ as an isometric subgraph. Then, for four extreme vertices $x$, $y$, $z$, $t$, we  have
    $d(x,y)=d(y,z)=d(z,t)=d(t,x)=k+2$ and $d(x,z)=d(y,t)=2k+3$.
    Thus, $x$, $y$, $z$, $t$  form $C_4$ in $G^\ell$ for all $\ell \in [k+2, 2k+2]$.

    Next, let $x,y,z,t$ be four vertices in $G$ that form $C_4$ in $G^\ell$ for all $\ell \in [k+2, 2k+2]$.
    Then, each of $d(x,y)$, $d(y,z)$, $d(z,t)$, and $d(t,x)$ is less than or equal to $k+2$, and each of $d(x,z)$ and $d(y,t)$ is
    greater than or equal to $2k+3$. Additionally, $d(x,y)$, $d(y,z)$, $d(z,t)$, and $d(t,x)$ must be greater than or equal to $k+1$,
    since otherwise $d(x,z) < 2k+3$ and $d(y,t) < 2k+3$. Thus, $2k+3\leq d(x,z)\leq 2k+4$ and $2k+3\leq d(y,t)\leq 2k+4$.
    We consider three cases.

      In case 1, let $d(x,z)=d(y,t)=2k+4$. Then, necessarily, $d(x,y)=d(y,z)=d(t,z)=d(t,x)=k+2$. By Lemma \ref{s-frame-hellified},
      $G$ has an isometric $H_1^{k+2}$. Since $H_1^{k+1}$ is an isometric subgraph of $H_1^{k+2}$, $G$ has an isometric $H_1^{k+1}$.

      In case 2, let $d(x,z)=2k+3$ and $d(y,t)=2k+4$. Then $d(x,y)=d(y,z)=d(t,z)=d(t,x)=k+2$ (otherwise, $d(y,t) < 2k + 4$).
      As in the proof of Lemma \ref{int-graph-power}, we conclude that $G$ has an isometric $H_2^{k+1}$. Thus, $G$ has an isometric $H_1^{k+1}$ (recall that $H_2^{k+1}$ contains an isometric $H_1^{k+1}$).

      In case 3, let $d(x,z)=d(y,t)=2k+3$. First assume, without loss of generality, that  $d(x,y)=k+1$. Then, necessarily, $d(x,t)=d(y,z)=k+2$.
      If also $d(t,z)=k+1$ then, by Lemma \ref{s-frame-hellified}, $G$ has an isometric $H_1^{k+1,k+2}$. Since $H_1^{k+1}$ is an isometric subgraph of $H_1^{k+1,k+2}$, $G$ has an isometric $H_1^{k+1}$.
      If now $d(t,z)=k+2$ then, by Lemma \ref{fnew} applied to $y,z,t$, there exist shortest paths $P(z,y)$ and $P(z,t)$ such that the neighbors of $z$ on those paths are adjacent.
      Let $z'$ be the neighbor  of $z$ on $P(z,y)$. We have $d(x,t)=d(t,z')=k+2$, $d(x,y)=d(y,z')=k+1$ and hence $d(x,z')=2k+2$ as $d(x,z)=2k+3$. 
      By Lemma \ref{fnew} applied to $x,z',t$, there exists a vertex $t'$ adjacent to $t$ such that $d(t',x)=k+1$ and $d(t',z')=k+1$,
      \heather{as shown in Fig.~\ref{fig:graphPowers}(a).}
      Since $d(x,y)=d(x,t')=k+1$ and $d(y,t)=2k+3$, necessarily $d(y,t')=2k+2$.
      By Lemma \ref{s-frame-hellified}, $G$ has an isometric $H_1^{k+1}$ with $x,y,z',t'$ as corner points. 

      To finish case 3, it remains to analyze the situation when $d(x,z)=d(y,t)=2k+3$ and $d(x,y)=d(y,z)=d(z,t)=d(t,x)=k+2$.
      By Lemma \ref{fnew} applied to $x,z,t$, there is a triangle $\triangle(t_x,t,t_z)$ such that $t_x$ is the neighbor  of $t$ on a shortest $(x,t)$-path, 
      and $t_z$ is the neighbor  of $t$ on a shortest $(z,t)$-path. 
      Similarly, by Lemma \ref{fnew} applied to $x,z,y$, there is a triangle $\triangle(y_x,y,y_z)$ such that $y_x$ is the neighbor of $y$ on a shortest $(x,y)$-path, 
      and $y_z$ is the neighbor of $y$ on a shortest $(z,y)$ path. 
      From the distance requirements, $2k+1 \leq d(t_x,y_x) \leq 2k+2$ and $2k+1 \leq d(t_z,y_z) \leq 2k+2$ (recall that $d(x,z)=d(y,t)=2k+3$).

      If $d(t_x,y_x)=2k+2$ then, by Lemma \ref{fnew} applied to $y_x,z,t_x$, there exists a vertex $z'$ adjacent to~$z$ such that $d(z',y_x)=d(z',t_x)=k+1$, \heather{as shown in Fig.~\ref{fig:graphPowers}(b).}
      Necessarily, $d(x,z')=2k+2$ as $d(x,z)=2k+3$. Now, $d(x,y_x)=d(x,t_x)=d(z',y_x)=d(z',t_x)=k+1$ and $d(y_x,t_x)=d(x,z')=2k+2$, 
      and we can apply Lemma \ref{s-frame-hellified} and get in $G$ an isometric $H_1^{k+1}$ with $x,y_x,z',t_x$ as corner points. 
      Thus, we may assume that $d(t_x,y_x)=2k+1$. Similarly, we may assume that $d(t_z,y_z)=2k+1$.
    
      By Lemma \ref{fnew} applied to $t_x,x,y_x$, there exist shortest paths $P(y_x,x)$ and $P(t_x,x)$ such that the neighbors of $x$ on those paths are adjacent. By Lemma \ref{fnew} applied to $y_z,z,t_z$, there exist shortest paths $P(y_z,z)$ and $P(t_z,z)$ such that the neighbors of $z$ on those paths are adjacent, \heather{as shown in Fig.~\ref{fig:graphPowers}(c).}
      Thus, we have constructed an $\{x,y,z,t\}$-distance preserving subgraph depicted on Fig. \ref{fig:unhellified-frames}(c). Hence, by Lemma \ref{f-frame-hellified}, $G$ has an isometric subgraph $H_3^k$.   \qed
  \end{proof}

   \begin{center}
   \centering
   \image{}{
     \includegraphics[scale=0.81]{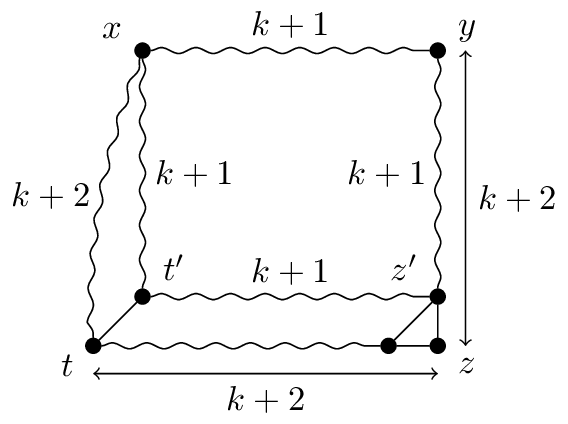}
   }\quad\image{}{
     \includegraphics[scale=0.81]{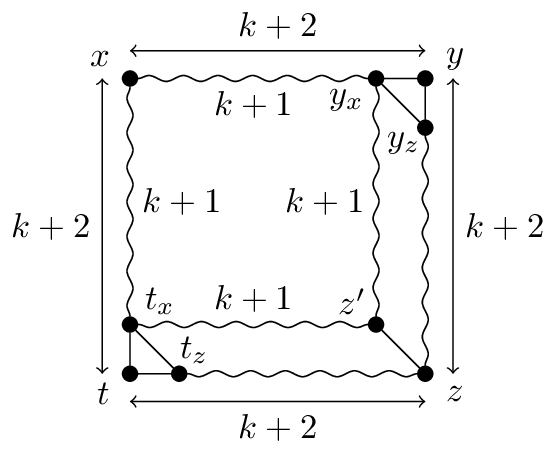}
   }\quad\image{}{
     \includegraphics[scale=0.81]{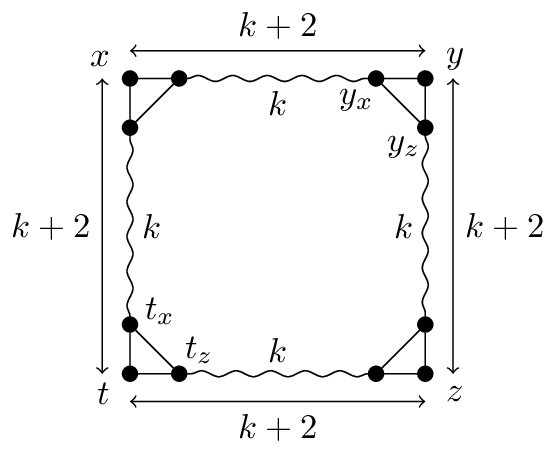}
   }\captionof{figure}{Illustrations for Case 3 for the proof of Lemma~\ref{half-int-graph-power}.}\label{fig:graphPowers}
   \end{center}\setcounter{subfloat}{0}

  The following result reformulates Theorem \ref{general-forbidden} in terms of graph powers. It follows directly from Theorem \ref{general-forbidden},
  Lemma \ref{int-graph-power}, and Lemma \ref{half-int-graph-power}.
  \heather{It relates to a result of Coudert and Ducoffe \cite{Coudert14}, which characterizes any $\frac{1}{2}$-hyperbolic graph by forbidding $C_4$ in certain graph powers.
  Here, we give a characterization for any $\delta$-hyperbolic Helly graph, for all values of $\delta$,
  by forbidding $C_4$ and the diamond graph in certain graph powers.}

  \begin{theorem} \label{thm:powers}
    Let $G$ be a Helly graph and $k$ be a non-negative integer. \\
    - $hb(G) \leq k$ if and only if
    there are no four vertices that form $C_4$ in $G^\ell$ for all $\ell \in [k+1, 2k]$ and form a diamond in $G^{2k+1}$. \\
    - $hb(G) \leq k+\frac{1}{2}$ if and only if
    there are no four vertices that form $C_4$ in $G^\ell$ for all $\ell \in [k+1, 2k+1]$, and
    there are no four vertices that form $C_4$ in $G^\ell$ for all $\ell \in [k+2, 2k+2]$.
  \end{theorem}

\medskip
\noindent{\bf Acknowledgment:} We would like to
thank anonymous reviewers for many useful suggestions and comments. 

 %
 %
 \bibliographystyle{plain}

\end{document}